\documentclass[12pt]{article}

\usepackage{authblk}
\usepackage{amssymb,amsmath}
\usepackage{bbm}
\usepackage{mathrsfs}
\usepackage{mathdots}

\usepackage[T1]{fontenc}
\usepackage{graphicx}
\usepackage{setspace}
\usepackage{natbib}
\usepackage{hyperref}

\usepackage[multiple]{footmisc}

\usepackage{tikz}
\usepackage{pgfplots}

\usetikzlibrary{arrows.meta}

\tikzstyle{block} = [draw, rectangle, minimum height=2em, minimum width=4em] \tikzstyle{sum} = [draw, circle, node distance=1cm]
\tikzstyle{dot} = [draw, fill=black, circle, inner sep=0.05cm]
\tikzstyle{input} = [coordinate] \tikzstyle{output} = [coordinate]
\tikzstyle{pinstyle} = [pin edge={to-,thin,black}]

\usepackage[margin= 1.0 in]{geometry}

\newtheorem{theorem}{Theorem}[section]
\newtheorem{lemma}[theorem]{Lemma}
\newtheorem{proposition}[theorem]{Proposition}
\newtheorem{corollary}[theorem]{Corollary}

\newenvironment{proof}[1][Proof:]{\begin{trivlist}
\item[\hskip \labelsep {\bfseries #1}]}{\end{trivlist}}

\newcommand{\qed}{\hfill \nobreak \ifvmode \relax \else 
      \ifdim\lastskip<1.5em \hskip-\lastskip
      \hskip1.5em plus0em minus0.5em \fi \nobreak
      \vrule height0.75em width0.5em depth0.25em\fi}

\newcommand\numberthis{\addtocounter{equation}{1}\tag{\theequation}}

\pgfplotsset{compat=1.14}

\begin{document}

\title{\textbf{The Origin and the Resolution\\ of Nonuniqueness\\ in Linear Rational Expectations}} 

\author{John G. Thistle\thanks{Department of Electrical and Computer Engineering, 
            University of Waterloo, 
            Waterloo, Ontario, 
            Canada\ \ N2L 3G1\ 
           \texttt{jthistle@uwaterloo.ca}.  
This paper is a revised version of arXiv:1806.06657 [q-fin.EC], which is in turn a revised version of Technical Report \#~2018-01 (April, 2018) of the author's academic department.  
It is a pleasure to thank Andrew Heunis and Christopher Nielsen for some stimulating discussions; Christopher Nielsen also offered comments on an earlier draft of the manuscript.  }} 

\renewcommand{\today}{}

\maketitle 

\begin{abstract}
The nonuniqueness of rational expectations is explained:  in the 
stochastic, discrete-time, linear, constant-coefficients case, 
the associated free parameters are coefficients that determine the public's most immediate reactions to shocks.  
The requirement of model-consistency may leave these parameters 
completely free, yet when their values are appropriately specified, 
a unique solution is determined.  
In a broad class of models, the requirement of least-square forecast errors determines 
the parameter values, and therefore defines a unique solution.  
This approach is independent of dynamical stability, and generally does not suppress model dynamics.  

Application to a standard New Keynesian example shows that the traditional solution 
suppresses precisely those dynamics that arise from rational expectations.  
The uncovering of those dynamics reveals their incompatibility 
with the new I-S equation and the expectational Phillips curve.  

\begin{flushleft}
\end{flushleft}
\end{abstract}

\section{Introduction}  

The rational expectations hypothesis plays a central role in macroeconomics, 
but it has long been recognized that it need not determine expectations uniquely, 
and the resulting indeterminacy has raised objections on philosophical, mathematical, and practical grounds.  
This paper traces nonuniqueness to its economic origins, 
in the dynamics of expectation formation.

It has in effect been widely supposed that the 
heart of the problem of nonuniqueness might lie in an infinite regression 
(if it is possible to speak of ``regression'' into the future):  
if the values of endogenous variables depend on contemporary forecasts of their future values, 
then those forecasts of future values depend on future forecasts, which in turn depend on 
still later forecasts; and so forth.  
This apparent regression has suggested a need for a terminal boundary condition, and raised the hope that 
such a condition might help in determining a unique solution.  
But this paper shows that the origin of indeterminacy lies not in the infinitely remote future, 
but in the immediate present:  
in the stochastic, linear, discrete-time, constant-coefficients case, 
the associated free parameters determine the immediate response of expectations to shocks.   
The requirement of rationality -- or model-consistency -- generally leaves those parameters 
unconstrained; yet when their values are appropriately specified, a unique solution is determined.  

In particular, if the requirement of model-consistency is strengthened with the additional requirement that 
forecasting errors be minimized in the least-squares sense, then, for a broad class of models, 
the free parameters are determined, and solutions are unique.  

\citet{RePEc:oup:oxford:v:34:y:2018:i:1-2:p:43-54.} has expressed the view that rational expectations are ``insufficiently inertial'':   
it is shown here that the traditional approach to rational expectations, in producing a unique solution to a standard  
New Keynesian model, 
suppresses exactly those eigenvalues that arise from rational expectations.  
In contrast, the methods of this paper do not presuppose any stability or instability properties, nor do they 
generally suppress any model dynamics.  
The uncovering of the dynamics of rational expectations in the New Keynesian model reveals incompatibility with the 
new I-S equation and the expectational Phillips curve.  

\subsection{Taylor's example}  
\label{Taylor}

The nature of the free parameters that underlie nonuniqueness is illustrated by an early example of \citet{RePEc:ecm:emetrp:v:45:y:1977:i:6:p:1377-85}.  
Consider the following equation, 
\begin{align*}
\hat{p}_{2,t-1} & = \hat{p}_{1,t-1} + \delta_1 p_t + u_t 
\end{align*} 
where $\/\delta_1\/$ is a nonzero real constant, $\/u_t\/$ is a sequence of independent, identically distributed, zero-mean random variables with finite variance, 
$\/p_t\/$ is an endogenous variable (specifically, proportional to the logarithm of the price of output)\footnote{This version of Taylor's equation has been linearized through a change of coordinates:  here, $\/p_t\/$ denotes $\/p_t - \delta_0/\delta_1\/$ in Taylor's coordinates.}, and, for $\/i = 1,2\/$, $\/\hat{p}_{i,t}\/$ denotes a forecast of $\/p_{t+i}\/$, formulated at time $\/t\/$ (that is, formulated in terms of $\/u_\tau\/$, 
for $\/\tau \leq t\/$).  

Seeking solutions of the form 
\begin{align*}
p_t = \sum_{i=0}^\infty \pi_i u_{t-i}\ , 
\end{align*} 
Taylor imposes rational expectations by setting 
\begin{align*}
\hat{p}_{1,t-1} & = \sum_{i=1}^\infty \pi_i u_{t-i}\ , \text{ and} \\
\hat{p}_{2,t-1} & = \sum_{i=2}^\infty \pi_i u_{t+1-i} \ .  
\end{align*}  
Substituting these expressions into the model equation, he finds:  
\begin{align*}
\pi_0 & = - \delta_1^{-1} \ , \\
\pi_{i+1} & = (1 + \delta_1) \pi_i \ ,\ \forall i \geq 1\ .  
\end{align*} 
The coefficient $\/\pi_0\/$ is determined,\footnote{The analysis of appendix~\ref{multexp} shows that $\/\pi_0\/$ is determined because the model equation contains no unlagged forecasts.} but $\/\pi_1\/$ is free, and its value determines all other coefficients.  
Taylor then turns to the imposition of a finite-variance condition (dynamical stability) 
and a minimum-variance condition, as means of 
limiting the possible values of $\/\pi_1\/$.  

But what is the economic significance of the quantity $\/\pi_1\/$?  
According to the parameterization of $\/\hat{p}_{1,t-1}\/$, 
\begin{align*}
\hat{p}_{1,t-1} & = \pi_1 u_{t-1} + \pi_2 u_{t-2} + \ldots\ ;  
\end{align*}
so the coefficient $\/\pi_1\/$ determines the effect of 
$\/u_{t-1}\/$ on this forecast.  
In other words, this coefficient determines the immediate effect of shocks; 
it models the integration of new information into the forecast.  
This is an undeniably important parameter of the ``expectations mechanism.'' 

Yet, $\/\pi_1\/$ is completely free under the assumption of rational expectations.  
As strong an assumption as rational expectations is, it does not determine, 
either in whole or in part, the integration of new information into the forecast.  
The results of this paper generalize this finding; they show that 
the free variables that account for the nonuniqueness of rational expectations
are precisely the coefficients that govern the immediate effects of shocks on forecasts.  
Consequently, nonuniqueness is resolved if and only if those immediate effects 
are modeled unambiguously -- by other means.  

The prevailing approach is to constrain such effects indirectly, principally by imposing 
a terminal condition, which requires dynamical stability.  
The use of the stability criterion is arbitrary, because it bears no particular relationship to the 
cause of nonuniqueness.  
It is also unrealistic, because it depends on 
infinite-precision cancellation of unstable dynamics.  
Taylor shows that, if $\/\delta_1\/$ is positive, then his model is stable  
if and only if $\/\pi_1\/$ is exactly zero.  
He derives the following recurrence:  
\begin{align*}
p_t - (1 + \delta_1) p_{t-1} & = - \delta_1^{-1} u_t + (\pi_1 + (1 + \delta_1) \delta_1^{-1} ) u_{t-1}\ .  
\end{align*} 
The moving averages in $\/p_t\/$ and $\/u_{t}\/$ are of the same form if $\/\pi_1 = 0\/$, leading to 
a so-called `pole-zero cancellation' that suppresses the unstable eigenvalue at $\/1 + \delta_1\/$.  
The cancellation is displayed more explicitly 
by bringing in a `left-shift' operator $\/z\/$, so that $\/z x_t\/$ stands for $\/x_{t+1}\/$, 
and $\/z^{-1} x_t\/$ for $\/x_{t-1}\/$ (so $\/z^{-1}\/$ is a `right-shift' or `lag' operator).    
Then, assuming that normal algebraic operations can be applied to the operator,\footnote{This operator notation is formalized with the use of the z-transform for the purposes of the rest of the paper.} the above equation can be rewritten as, 
\begin{align*}
p_t  & = - \delta_1^{-1} \frac{z - (\pi_1 \delta_1 + (1 + \delta_1))}{z - (1 + \delta_1)} u_t \ .  
\end{align*} 
Hence, if $\/\pi_1 = 0\/$, 
\begin{align*}
p_t & = - \delta_1^{-1} \frac{z - (1 + \delta_1)}{z - (1 + \delta_1)} u_t = - \delta_1^{-1} u_t\ .
\end{align*} 
The zero of the rational function in $\/z\/$ coincides with (and cancels) 
the unstable pole (the root of the denominator polynomial, which coincides with the characteristic polynomial) 
if and only if $\/\pi_1 = 0\/$.  
In that case, the resulting model has reduced-order dynamics -- and therefore, no dynamics at all.  

This method of ensuring a unique solution depends on the model's possessing an appropriately unstable eigenvalue.  
It also requires the immediate response of the public's aggregate forecast $\/\hat{p}_{1,t}\/$ to the input $\/u_t\/$ to equal zero, exactly.  

Another way of specifying the same unique value for $\/\pi_1\/$ in Taylor's example is to require in addition that the variance of the forecast errors be minimized.  
This paper shows that, in a broad class of models, that requirement determines a unique solution, regardless of stability properties, 
and generally without entailing pole-zero cancellation.  

\subsection{Background}  

The source of nonuniqueness of rational-expectations models lies in the reason for  
modeling expectations in the first place:   
forecasts have a bearing on the behavior of 
economic variables; 
moreover, they may affect the very quantities being forecast.  
The study of this self-referential phenomenon 
dates at least to \citet{Tinbergen:horizon}, who examined  
the effect of forecast horizons on the movement of commodity prices.  
The work of \citet{RePEc:ucp:jpolec:v:62:y:1954:p:465}, on the public prediction of social 
events, makes the circularity of the problem explicit, 
treating model-based expectations as fixed points.  
This property was later summarized 
by \citet{RePEc:eee:moneco:v:4:y:1978:i:1:p:1-44} as 
that of an ``expectations mechanism which `reproduces itself' in [the] model.''  
The fixed-point characterization immediately raises the possibilities of nonexistence and nonuniqueness.  

\citet{Muth:rational} applied similar ideas to market dynamics, hypothesising that there were no 
systematic discrepancies between the forecasts  of market participants and the predictions of economic theory.  
He called such 
forecasts {\em rational expectations\/}.\footnote{\citet{RePEc:ecj:econjl:v:101:y:1991:i:408:p:1245-53} compares Muth's contribution to Tinbergen's.}  
\citet{RePEc:eee:jetheo:v:4:y:1972:i:2:p:103-124,RePEc:eee:crcspp:v:1:y:1976:i::p:19-46} applied the rational expectations hypothesis 
to macroeconomic models, to show how changes in anticipated policy may give rise to 
changes in the behavior of economic agents.  
The advent of rational expectations is often described as a revolution in macroeconomics, 
and the approach has been developed by many economists, of different schools of thought.  
But as soon as it was applied to models with significant dynamics, and corresponding 
forecasts of the future outcomes of those dynamics, it was found that expectations were 
not defined uniquely \citep{RePEc:ecm:emetrp:v:45:y:1977:i:6:p:1377-85,RePEc:eee:moneco:v:4:y:1978:i:1:p:1-44};\footnote{\citet{RePEc:eee:jetheo:v:7:y:1974:i:1:p:53-65} focused on nonuniqueness, but under the relatively strong requirement that the initial conditions consistent with continual economic equilibrium, and a bounded rate of inflation, be unique.} 
the potential for nonuniqueness inherent in Grunberg and Modigliani's  
fixed-point characterization was realized.  

In the absence of an explanation of the source of this nonuniqueness, 
the predominant response has been to impose a terminal boundary condition, 
limiting the growth rates of trajectories, in the hope that it is satisfied by a unique solution.  
Indeed, the term ``solution'' has almost come to mean ``stable solution'' \citep{Funovits201747}.  
If the underlying model has appropriate unstable eigenstructure, its instabilities may  
be sufficient to narrow the space of stable solutions to a singleton.\footnote{See, for example, \citet{RePEc:ecm:emetrp:v:41:y:1973:i:6:p:1043-48,RePEc:eee:moneco:v:4:y:1978:i:1:p:1-44,MINFORD1979117,RePEc:ecm:emetrp:v:48:y:1980:i:5:p:1305-11,RePEc:cup:etheor:v:13:y:1997:i:06:p:877-888_00,RePEc:ier:iecrev:v:39:y:1998:i:4:p:1015-26,RePEc:kap:compec:v:20:y:2002:i:1-2:p:1-20,RePEc:aea:aecrev:v:94:y:2004:i:1:p:190-217}.}
This approach entails the cancellation of unstable dynamics in one fashion or another, 
which effectively requires the public to act, in aggregate, with infinite precision.   
The effect of the cancellation is to obliterate dynamical features of the model.   
The result of these practices has been that proponents and critics alike have confined their studies 
of rational expectations to pathologically special instances.   
Moreover, these ad hoc methods may still be insufficient to resolve nonuniqueness: for such cases, 
a variety of other ideas, such as minimum-variance solutions, or minimal 
state-variable realizations, have also been proposed \citep{RePEc:ecm:emetrp:v:45:y:1977:i:6:p:1377-85,BASAR1989591,RePEc:kap:itaxpf:v:6:y:1999:i:4:p:621-639,EH:learning}.  
However, like dynamical stability, such criteria fail to get to the heart of the matter.  

\subsection{Overview}  

The main point of the article is to explain nonuniqueness by identifying the associated free parameters, 
which turn out to be parameters of the expectations mechanism itself.  
For the sake of generality, the analysis proceeds exclusively from simple, minimal assumptions, 
that are satisfied by traditional approaches to rational expectations.  
Beyond a standard, technical assumption, it is supposed only that forecasts depend linearly on initial conditions and on the model's driving variables, 
their dependence on the driving variables being representable by means of a linear difference equation with 
constant coefficients.  
This mild assumption echoes assumption 3 of \citet{Muth:rational};  
it is satisfied by conventional rational-expectations solutions; and it is necessary for the preservation of the linear, constant-coefficient structure of the model equations.  
In this paper, it is formalized in the form of equations called {\em forecasting mechanisms\/}.   

The problem of deriving a general rational-expectations solution consists in solving 
for a forecasting mechanism subject to the constraint of model-consistency -- that is, subject to the constraint 
that the forecasting mechanism ``reproduce itself in the model.''   
That problem in turn reduces to the solution of some generally singular, 
deterministic matrix difference equations that describe the interrelationships of model 
parameters.  
Because these equations do not lend themselves to a simple recursion either forward or backward in time, 
it is useful to solve them in the frequency domain, 
by means of the z-transform (and without loss of generality).  
Uniqueness demands only the appropriate specification of a parameter of the forecasting mechanism that 
governs the immediate response of forecasts to shocks; 
it is that parameter that distinguishes one fixed-point forecasting mechanism from another.  
In other words, it is that parameter that distinguishes different rational-expectations solutions.  

The determination of the necessary form of any 
model-consistent forecasting mechanism (in sections~\ref{zsresp} and \ref{ziresp}) 
leads to a necessary and sufficient condition for existence, and to 
a characterization of the general solution in terms of the aforementioned parameter (section~\ref{totresp}).  
This result explains the nature of the nonuniqueness of rational-expectations models, 
with reference to the dynamics of expectation formation.  
Because of the minimal set of assumptions on which it is based, 
all of this analysis is general:  it encompasses any reasonable approach to rational expectations under 
linear, constant-coefficient models.  

This general picture can be tidied up through the assumption of a simple structural 
condition, called well-posedness.  
Well-posedness ensures existence of a model-consistent forecasting mechanism, and permits the realization 
of that forecasting mechanism in the relatively robust form of a predictor that incorporates feedback (section~\ref{wellposed}).  

Conventional stability criteria, as ad hoc means of consuming the undesired degrees of freedom in the general 
solution, are antithetical to the spirit of this paper, but it is nevertheless shown for purposes of comparison (section~\ref{example}) 
that the framework of this paper can reproduce well known observations on the ``determinacy'' of a small New Keynesian model.  
\citet{RePEc:oup:oxford:v:34:y:2018:i:1-2:p:43-54.} has expressed the view that rational expectations are insufficiently inertial:  
it turns out that the pole-zero cancellation required for stability suppresses precisely those dynamics of the New Keynesian model 
that arise from expectations (section~\ref{stabilitysubs}).  

For a broad class of models, nonuniqueness can be eliminated simply by taking the assumption of rational expectations -- 
or unbiased forecasts -- one step further, and assuming that forecast errors are not only zero-mean, but are minimized, in the least-squares sense 
(section~\ref{leastsquares}).  
The strengthened assumption is arguably milder than that of unstable pole-zero cancellations, 
and does not require any particular stability or instability properties of the model.  
Nor does it generally result in the cancellation of dynamics.  

Because the methods of the paper are independent of dynamical stability considerations, 
 they allow for the study of stabilization via policy.  
This point is illustrated in section~\ref{stabilization}, 
where it emerges that the standard formulation of the New Keynesian model is incompatible with 
the dynamics of rational expectations.  

A brief review is given in section~\ref{related} of the vast related literature, and the paper concludes 
with some general suggestions for research.   
Appendices show how key results of the paper extend to more general models,  
give details of some proofs omitted from the main body of the paper,  
and outline relevant mathematical background, for consultation as necessary.  

\section{Problem formulation}  
\label{statement}

The exposition is based on a simple, abstract model of \citet{CHO2015160}, 
though the principles generalize (see appendix~\ref{multexp}).  

For all $\/t \in \mathbb{Z}\/$, 
\begin{align}
\label{stateqn}
x_{t} & = A x_{t-1} + \hat{A} \hat x_{1,t} + B u_{t} \\ 
\label{inputeqn}
u_{t} & = R u_{t-1} + w_{t} \ .  
\end{align}
The matrices 
$\/A,\hat{A} \in \mathbb{R}^{n \times n}\/$ ($\/\hat{A} \neq 0\/$), $\/B \in \mathbb{R}^{n \times m}\/$ are constants.

The variables include the independent variable $\/t \in \mathbb{Z}\/$, representing discrete time instants, 
a vector of endogenous variables $\/x_t \in \mathbb{R}^{n \times 1}\/$, 
the vector $\/\hat{x}_{1,t} \in \mathbb{R}^{n \times 1}\/$ representing a one-time-step ``forecast'' of the value of $\/x_t\/$, 
and a vector of exogenous inputs $\/u_t \in \mathbb{R}^{m \times 1}\/$, 
driven by a sequence $\/w_t \in \mathbb{R}^{m \times 1}\/$ of real-valued, independent, zero-mean random variables, 
with finite variance,\footnote{Specifically, all entries of the covariance (or variance) matrix are finite.} 
defined on a common probability space.  
The forecast $\/\hat{x}_{1,t}\/$ may depend on the initial conditions $\/x_{-1}\/$, $\/\hat{x}_{1,-1}\/$ and $\/u_{-1}\/$, 
which are formally treated as constants, and on the sequence of random variables $\/u_0\/$, $\/u_1\/$,\ \ldots\ , $\/u_t\/$ 
-- or equivalently, the sequence $\/w_0, w_1,\ \ldots\ , w_t\/$.  

It will be assumed that the polynomial matrix 
 $\/[z^2 \hat{A} - z I + z]\/$ is {\em regular\/} 
-- meaning that its determinant does not vanish for all $\/z \in \mathbb{C}\/$, 
in which case the model too will be called regular.  
This is a common assumption, which serves in this instance to rule out a source of nonuniqueness that is unrelated to expectations.\footnote{The model is intended by Cho and McCallum to capture a local, linear approximation around the steady state of a dynamic, stochastic, general-equilibrium (DSGE) model.  For a concrete, numerical instance, see section~\ref{example}.  Under the conventional rational-expectations paradigm, the Cho-McCallum model has been used to capture equations with arbitrary finite numbers of expectation terms, having arbitrary expectational leads and lags \citep{broze_gouriŽroux_szafarz_1995,RePEc:cup:etheor:v:13:y:1997:i:06:p:877-888_00,RePEc:eee:dyncon:v:31:y:2007:i:4:p:1376-1391}; see Appendix~\ref{multexp} for a direct extension of the results of the paper to such general models.  }    

Of course, the random variables $\/\hat{x}_{1,t}\/$ and $\/x_t\/$ are undefined, 
in the absence of additional equations.  
For the sake of generality, it will be assumed only that the forecasts $\/\hat{x}_{1,t}\/$ depend linearly 
on the initial conditions and the driving variables $\/w_t\/$ -- and moreover, that their dependence on the driving 
variables can be represented by a linear, constant-coefficient, stochastic difference equation.  
This is in essence a version of assumption 3 of \citet{Muth:rational}, and it is satisfied by other approaches to rational expectations.  
The theory of linear, constant-coefficient difference equations then tells us that 
expectations must then obey an equation of the following form:  
\begin{align*}
\hat{x}_{1,t} = \sum_{\tau = 0}^t \tilde{F}_{t - \tau} w_\tau + \overline{x}_{t+1}\ , 
\numberthis \label{tildeforecasts} 
\end{align*} 
where $\/\tilde{F}_t \in \mathbbm{R}^{n \times m}\/$ vanishes for negative $\/t\/$, 
and $\/\overline{x}_{t+1}\/$ depends linearly on the initial conditions, but does not depend on the $\/w_t\/$.  
Such an equation will be called a {\em forecasting mechanism\/}.  
The convolution kernel determines the 
matrix of impulse responses of the forecasting mechanism, and the limits of the convolution sum ensure that forecasts are based on the appropriate 
information set.  

The variables $\/w_t\/$ are in essence merely a convenient means of defining the 
stochastic structure of the $\/u_t\/$,  
so it is also of interest to consider forecasting mechanisms driven by the 
economic, exogenous variables $\/u_t\/$:  
\begin{align*}
\hat{x}_{1,t} = \sum_{\tau = 0}^t F_{t - \tau} \tilde{u}_\tau + \overline{x}_{t+1}\ ,\ \forall t \geq 0, 
\numberthis \label{forecasts} 
\end{align*} 
For every $\/t \geq 0\/$,  $\/F_t \in \mathbbm{R}^{n \times m}\/$.  
The sequence $\/\tilde{u}_t := \sum_{\tau = 0}^t R^{t-\tau} w_\tau = u_t - R^{t+1} u_{-1}\/$ denotes the component of the sequence of endogenous variables 
$\/u_t\/$ that depends only on the $\/w_t\/$ and not on the initial conditions.  

All of the analysis carried out in the next section follows from the above mild assumption on the form of forecasts, 
together with the aforementioned regularity property.  
It therefore represents a general analysis of rational expectations, 
within the context of linear, constant-coefficient, stochastic difference equations.  

Of specific interest are forecasting mechanisms that are unbiased.  
Let $\/y\/$ be any square-integrable random variable defined on the common probability space of the $\/w_t\/$.  
For $\/t \geq -1\/$, $\/E_t(y)\/$ denotes the expected value of $\/y\/$, conditioned on 
the driving variables $\/w_0\/$, $\/w_1\/$,\ \ldots\ , $\/w_t\/$ 
-- and subject to the full model, including the forecasting mechanism, and its initial conditions.  
Given consistent initial conditions, a forecasting mechanism~(\ref{tildeforecasts}) (respectively,~(\ref{forecasts})) is {\em model-consistent\/} if 
the full model~(\ref{stateqn}--\ref{tildeforecasts}) (respectively,~(\ref{stateqn},\ref{inputeqn},\ref{forecasts})) satisfies 
\begin{align*}
E_t(x_{t+1} - \hat{x}_{1,t}) = 0 \ \mbox{, or equivalently,}\ \hat{x}_{1,t} = E_t(x_{t+1})\ , \forall t \geq -1\ .  
\end{align*} 
Such a forecasting mechanism embodies the ``strong'' rational-expectations hypothesis, whereby economic agents behave, in aggregate, 
as if they have access to all relevant information about the economy, and, on that basis, form expectations that do not 
incorporate any systematic errors.  

It should be emphasized that the appropriate conditional expectations are subject to the forecasting mechanism, 
even if that forecasting mechanism is initially unknown to the modeler or analyst.  
This is a crucial point, which the author considers to be logically implied by, 
and entirely in the spirit of, the strong rational-expectations hypothesis.  
In economic terms, it represents the assumption that economic actors (in their aggregate) 
behave as if they know not only how the economy responds to shocks and to actors' expectations, 
but also how actors' own expectations are formed, and 
revised in response to shocks.  

The main results of the paper include a general existence-and-uniqueness result for 
model-consistent forecasting mechanisms (Theorem~\ref{exun}, page~\pageref{exun}), 
and identification of a simple structural condition that ensures existence, as well as realizability in the form of  
combined feedforward/feedback implementations (Theorem~\ref{feedbackthm}, page~\pageref{feedbackthm}).  
Finally, it is shown that uniqueness can be ensured, for a broad class of models, by 
requiring not only that forecasting errors be zero-mean, but that they be minimized, in the least-squares sense 
(Corollary~\ref{leastsquarederror}, page~\pageref{leastsquarederror}).

\section{The general model-consistent solution} 

To derive a general representation of model-consistent forecasting mechanisms, we 
consider separately the case where the initial conditions are zero-valued, and that where 
the driving variables are zero-valued; by linearity, we then superimpose the resulting solutions.  

It is in the first case that the reason for nonuniqueness becomes apparent.  

\subsection{Zero-state response} 
\label{zsresp}

Suppose 
that the initial conditions $\/x_{-1},\/$ $\/\hat{x}_{1,-1}\/$, and $\/u_{-1}\/$ are zero.  
In the terminology of linear, time-invariant systems, the resultant response 
is called the {\em zero-state response\/}.  

By definition, the zero-state response of the forecasting mechanism must have the form of a convolution.  
It is convenient in the first instance to find a description in terms of the driving 
variables $\/w_t\/$:  
\begin{align}
\hat{x}_{1,t} = \sum_{\tau = 0}^t \tilde{F}_{t - \tau} w_\tau \ ,\ \forall t \geq 0\ .
\numberthis \label{fconv}
\end{align} 
Because the model~(\ref{stateqn},\ref{inputeqn}) is then effectively a system of linear, constant-coefficient equations, 
the sequence of endogenous-variable vectors must have the form of a similar convolution:  
\begin{align*}
x_t = \sum_{\tau = 0}^t \tilde{G}_{t-\tau} w_\tau \  , \ \forall t \geq 0\ .  
\numberthis \label{gconv} 
\end{align*}
Naturally, $\/\tilde{G}_t\/$ will depend on $\/\tilde{F}_t\/$, and vice versa.  
The model and the model-consistency condition will furnish two equations 
describing their relationship.  

Substitute the convolution sums~(\ref{fconv},\ref{gconv}) into the equation~(\ref{stateqn}), and use~(\ref{inputeqn}) to eliminate $\/u_t\/$:  
\begin{align*}
\sum_{\tau = 0}^{t} \tilde{G}_{t-\tau} w_\tau = A \sum_{\tau = 0}^{t-1} \tilde{G}_{t - 1 - \tau} w_\tau 
                                                                     + \hat{A} \sum_{\tau = 0}^{t} \tilde{F}_{t-\tau} w_\tau 
                                                                     + B \sum_{\tau = 0}^{t} R^{t-\tau} w_\tau \ \forall t \in \mathbb{Z}\ . 
\numberthis \label{subconv} 
\end{align*}
Applying the conditional expectation operator $\/E_0\/$,  
\begin{align*}
\tilde{G}_{t} w_0 = A \tilde{G}_{t-1} w_0 + \hat{A} \tilde{F}_{t} w_0 + B R^{t} w_0\ ,\ \forall t \geq 0 \ .  
\end{align*} 
But here, $\/w_0 \in \mathbb{R}^{m \times 1}\/$ is arbitrary, so it must be that 
\begin{align*}
\tilde{G}_{t} = A \tilde{G}_{t-1}+ \hat{A} \tilde{F}_{t}  + B R^{t}\ ,\ \forall t \geq 0 \ .  
\numberthis \label{gfrecur}
\end{align*}

For a second equation relating the two impulse responses, bring in the model-consistency condition for $\/t \geq 0\/$:  
\begin{align*}
 \hat{x}_{1,t} = E_t(x_{t+1}) &
\iff \sum_{\tau = 0}^t \tilde{F}_{t - \tau} w_\tau = E_t(\sum_{\tau = 0}^{t+1} \tilde{G}_{t+1-\tau} w_\tau)\ , \\ 
& \iff \sum_{\tau = 0}^t \tilde{F}_{t - \tau} w_\tau = \sum_{\tau = 0}^t \tilde{G}_{t+1-\tau} w_\tau\ .   
\end{align*} 
Once again taking expected values conditioned on $\/w_0\/$, and factoring out $\/w_0\/$ 
on the grounds that the resulting equation must hold for arbitrary $\/w_0\/$, 
\begin{align*}
\tilde{F}_t = \tilde{G}_{t+1} \ ,\ \forall t \geq 0\  .  
\numberthis \label{onestep}
\end{align*}

The calculation of the zero-state response amounts to solving the system~(\ref{gfrecur},\ref{onestep}), and 
thus effectively reduces to the solution of the equation obtained by substituting $\/\tilde{G}_{t+1}\/$ for $\/\tilde{F}_t\/$ in~(\ref{gfrecur}):  
\begin{align*}
\tilde{G}_t = A \tilde{G}_{t-1} + \hat{A} \tilde{G}_{t+1} + B R^t,\ t\geq 0\ .  
\numberthis \label{grecur} 
\end{align*} 
For this, it is assumed that the matrix polynomial $\/[z^2 \hat{A} - z I + A]\/$ 
is {\em regular\/} -- that its determinant is not identically zero.  
This is a standard assumption, for which there is ample justification.  
See, for example, \citep{RePEc:ier:iecrev:v:39:y:1998:i:4:p:1015-26,RePEc:eee:ecolet:v:61:y:1998:i:2:p:143-147}.  
In particular, it is necessary for the uniqueness of solutions, 
irrespective of expectations.  
Under such a regularity assumption, any solution that exists is unique.  
Moreover, it is of exponential order, 
and therefore possesses a unilateral z-transform:  
consider that~(\ref{grecur}) can be rewritten as a first-order matrix difference equation, 
via {\em linearization\/} of the above matrix polynomial~\citep{Lanc:lin}; 
and see, for example, the results of \citet{Bruell2009}.   

Applying the unilateral z-transform, 
and its left- and right-shift properties, 
to (\ref{gfrecur},\ref{onestep}), 
and keeping in mind that $\/\tilde{G}_t\/$ must vanish for negative $\/t\/$, 
one finds 
\begin{align*} 
\numberthis \label{gz}
\tilde{G}[z] & = A z^{-1} \tilde{G}[z] + \hat{A} \tilde{F}[z] + B [I - R z^{-1}]^{-1}\ ,\\
\tilde{F}[z] & = z[\tilde{G}[z] - \tilde{G}_0] \ .  
\numberthis \label{fz} 
\end{align*} 
The appearance of $\/\tilde{G}_0\/$ in the equations reflects the need for an additional boundary 
condition; nonuniqueness results  when, in effect, only the condition 
$\/\tilde{G}_{-1} = 0\/$ is applied.  
But, given the form of the model~(\ref{stateqn},\ref{inputeqn}), such nonuniqueness can only arise from that of 
model-consistent forecasting mechanisms.  
The additional boundary condition should therefore be expressed in terms of $\/\tilde{F}_t\/$:  
setting $\/t = 0\/$ in~(\ref{grecur}) and~(\ref{onestep}), %one indeed finds 
\begin{align*}
\tilde{G}_0 = \hat{A} \tilde{G}_1 + B = \hat{A} \tilde{F}_0 + B\  .   
\numberthis \label{freeparam}  
\end{align*}

Substituting $\/\hat{A}\tilde{F}_0 + B\/$ for $\/\tilde{G}_0\/$ then, 
solve for the transforms $\/\tilde{F}[z]\/$ and $\/\tilde{G}[z]\/$ in terms of $\/\tilde{F}_0\/$:    
\begin{align*}
\tilde{F}[z] 
                 & = [z^2 \hat{A} - z I + A]^{-1} \left [[zI - A](\hat{A}\tilde{F}_0 + B)[zI - R] - z^2 B \right ] [I - R z^{-1}]^{-1}\ , \\
\tilde{G}[z] 
                  & = [z^2 \hat{A} - z I + A]^{-1} \left [ z \hat{A} (\hat{A}\tilde{F}_0 + B) [zI - R] - z B \right ] [I - R z^{-1}]^{-1}\ .  
\end{align*}
This establishes that, if suitable solutions of~(\ref{gfrecur},\ref{onestep}) exist, 
they must have transforms of the above forms, and therefore 
must be unique.  
But the form of $\/\tilde{F}[z]\/$ also determines existence:  
 
\begin{proposition}
\label{tildefgsoln}

Suppose that $\/[z^2 \hat{A} - zI + A]\/$ is regular.   
Then for any given value of the product 
$\/\hat{A} \tilde{F}_0 \in \mathbbm{R}^{n \times m}\/$, 
there exists a solution 
of the system~(\ref{gfrecur},\ref{onestep}), such that $\/\tilde{F}_t\/$ and $\/\tilde{G}_t\/$ 
vanish for negative $\/t\/$, and $\/\hat{A}\tilde{F}_0\/$ has the specified value, 
if and only if the rational matrix 
\begin{align*}
\tilde{F}[z] = [z^2 \hat{A} - z I + A]^{-1} \left [[zI - A](\hat{A}\tilde{F}_0 + B)[zI - R] - z^2 B \right ] [I - R z^{-1}]^{-1} 
\end{align*} 
is proper.  

In that case, the inverse z-transforms 
\begin{align*}
\tilde{F}_t & :={\cal Z}^{-1}\left \{ [z^2 \hat{A} {-} z I {+} A]^{-1} \left [[zI {-} A](\hat{A}\tilde{F}_0 + B)[zI {-} R] - z^2 B \right ] [I {-} R z^{-1}]^{-1} \right \}\\  
\tilde{G}_t & := {\cal Z}^{-1}\left \{ [z^2 \hat{A} {-} z I {+} A]^{-1} \left [ z \hat{A} (\hat{A}\tilde{F}_0 + B) [zI {-} R] - z B \right ] [I {-} R z^{-1}]^{-1} \right \} 
\end{align*} 
comprise the unique such solution.  
\end{proposition}

\begin{proof}
Suppose that an appropriate solution of the system exists.  
Then $\/\tilde{F}_t\/$ and $\/\tilde{G}_t\/$ vanish for negative $\/t\/$.  
Their respective z-transforms are then proper rational matrices; 
and, by the preceding discussion, $\/\tilde{F}[z]\/$ and $\/\tilde{G}[z]\/$ have the form given in the statement of the proposition, 
where $\/\tilde{F}_0\/$ is the initial value of $\/\tilde{F}_t\/$.  
This establishes uniqueness (by z-transform inversion), and the necessary condition for existence.  

Conversely, for any specified value of $\/\hat{A} \tilde{F}_0\/$, 
if the given matrix $\/\tilde{F}[z]\/$ is proper, then so is  
\begin{align*}
\tilde{G}[z] & = [z^2 \hat{A} - z I + A]^{-1} [z^2 \hat{A} (\hat{A} \tilde{F}_0 + B) - z B [I - R z^{-1}]^{-1}] \\
       & = (\hat{A} \tilde{F}_0 + B) + [z^2 \hat{A} - z I + A]^{-1}[[zI {-} A] (\hat{A} \tilde{F}_0 + B) - z B [I {-} R z^{-1}]^{-1}] \\
       & = (\hat{A}\tilde{F}_0 + B) + z^{-1} \tilde{F}[z]\ .  
\end{align*} 
Both $\/\tilde{F}[z]\/$ and $\/\tilde{G}[z]\/$ are therefore unilateral z-transforms.  
Because $\/z^{-1} \tilde{F}[z]\/$ is then strictly proper, the initial value of the inverse transform 
of $\/\tilde{G}[z]\/$ must equal 
$\/\hat{A} \tilde{F}_0 + B\/$.  
It follows that~(\ref{gz}) and~(\ref{fz}) are satisfied.  
Transforming back to the time domain then shows that~(\ref{gfrecur}) and~(\ref{onestep}), and consequently~(\ref{grecur}), are satisfied.  

Applying~(\ref{grecur}) at $\/t=0\/$ gives $\/\tilde{G}_0 = \hat{A} \tilde{G}_{1} + B\/$, 
so $\/\hat{A} \tilde{G_1}\/$ equals the specified value of $\/\hat{A} \tilde{F}_0\/$.  
But by~(\ref{onestep}), $\/\tilde{F}_t = \tilde{G}_{t+1}\/$ for all $\/t \geq 0\/$.  
So the product $\/\hat{A}\tilde{F}_t\/$ indeed has the specified value at $\/t=0\/$.  

This proves the sufficient condition for existence.  \qed
\end{proof}

The above result leads to a necessary condition on the form of 
model-consistent forecasting mechanisms~(\ref{tildeforecasts}):  

\begin{corollary}
\label{tildezscor}
Let the model~(\ref{stateqn},\ref{inputeqn}) be regular.  
For any given value of the product $\/\hat{A} \tilde{F}_0 \in \mathbbm{R}^{n \times m}\/$, 
any model-consistent forecast mechanism~(\ref{tildeforecasts}) 
must satisfy 
\begin{align*}
\tilde{F}_t & ={\cal Z}^{-1}\left \{[z^2 \hat{A} {-} z I {+} A]^{-1} \left [[zI {-} A](\hat{A}\tilde{F}_0 + B)[zI {-} R] - z^2 B \right ] [I {-} R z^{-1}]^{-1} \right \}\ . 
\end{align*} 
with $\/\tilde{F}[z]\/$ proper.  
 
If $\/x_{-1}\/$, $\/\hat{x}_{1,-1}\/$, and $\/u_{-1}\/$ are zero, 
and $\/\hat{x}_{1,t} = \sum_{\tau = 0}^t \tilde{F}_{t - \tau} w_\tau\/$, 
then the model (\ref{stateqn},\ref{inputeqn}) satisfies 
\begin{align*}
x_t = \sum_{\tau = 0}^t \tilde{G}_{t - \tau} w_\tau\ , 
\end{align*} 
where 
\begin{align*} 
\tilde{G}_t & := {\cal Z}^{-1}\left \{ [z^2 \hat{A} - z I + A]^{-1} \left [ z \hat{A} (\hat{A}\tilde{F}_0 + B) [zI - R] - z B \right ] [I - R z^{-1}]^{-1} \right \} \ .
\end{align*} 
The forecast error realized at time $\/t \geq 0\/$ is $\/x_t - \hat{x}_{1,t-1} = (\hat{A} \tilde{F}_0 + B) w_{t}\/$.  
\end{corollary} 

\begin{proof} 
The necessary conditions on the form of model-consistent forecast mechanisms, and on the solutions of the resulting models, are 
direct consequences of Proposition~\ref{tildefgsoln}, by the foregoing discussion.  

For the satisfaction of the model equations, we have, again by Proposition~\ref{tildefgsoln}, that equation~(\ref{gfrecur}) is satisfied.  
By linearity, and the fact that $\/\tilde{G}_t\/$ vanishes for negative $\/t\/$, 
so then is~(\ref{subconv}).  
It follows that the $\/\hat{x}_{1,t}\/$ and the $\/x_t\/$ in the statement of the corollary solve the model equations~(\ref{stateqn},\ref{inputeqn}).  

For the forecast dated at $\/t=-1\/$, then, 
\begin{align*}
x_0 - \hat{x}_{1,-1} = x_0 = \tilde{G}_0 w_0 = (\hat{A} \tilde{F}_0 + B) w_0\ ; 
\end{align*} 
and by equation~(\ref{onestep}), for all $\/t \geq 0\/$,  
\begin{align*}
x_{t+1} - \hat{x}_{1,t} = \sum_{\tau = 0}^{t+1} \tilde{G}_{t+1 - \tau} w_\tau - \sum_{\tau = 0}^t \tilde{F}_{t - \tau} w_\tau  
                                     = \tilde{G}_0 w_{t+1} 
                                     = (\hat{A} \tilde{F}_0 + B) w_{t+1}\ .  
\numberthis \label{errorformula}
\end{align*} \qed
\end{proof} 

These forecast errors are zero-mean, confirming model-consistency.

The solution can be expressed in terms of the exogenous inputs $\/u_t\/$ rather than 
the driving variables $\/w_t\/$, by simply right-multiplying $\/\tilde{F}[z]\/$ and 
$\/\tilde{G}[z]\/$ by $\/[I - R z^{-1}]\/$:  
\begin{align*}
{F}[z] & = [z^2 \hat{A} - z I + A]^{-1} [[zI - A](\hat{A}\tilde{F}_0 + B)[zI - R] - z^2 B] \ , \\
{G}[z] & = [z^2 \hat{A} - z I + A]^{-1} z [ \hat{A} (\hat{A}\tilde{F}_0 + B)[zI - R] -  B] \ , 
\end{align*}
The time-domain counterpart of right-multiplying by $\/[I - R z^{-1}]^{-1}\/$ to recover $\/\tilde{F}[z]\/$ and $\/\tilde{G}[z]\/$ 
is convolution with $\/R^t\/$.  
It follows that $\/F_0 = \tilde{F}_0\/$.  
It also follows that convolution of $\/F_t\/$ (respectively, $\/G_t\/$) with $\/u_t\/$ 
is equivalent to convolution of $\/\tilde{F}_t\/$ (resp.,~$\/\tilde{G}_t\/$) with $\/w_t\/$ 
(by the associativity of convolution).  
Because $\/[I - R z^{-1}]\/$ and its rational-matrix inverse are both proper, 
and because the initial value of each of their inverse z-transforms is the identity matrix, 
multiplication by either of them does nothing to alter properness.  

The following counterpart of Corollary~\ref{tildezscor} is immediate:  

\begin{corollary}
\label{zscor}
Suppose that 
the model~(\ref{stateqn},\ref{inputeqn}) 
is regular.   
Then, for any specified value of $\/ \hat{A} F_0 \in \mathbbm{R}^{n \times m}\/$, 
any model-consistent forecasting mechanism~(\ref{forecasts}) has 
\begin{align*}
F_t = {\cal Z}^{-1}\{F[z]\} = {\cal Z}^{-1}\{[z^2 \hat{A} - z I + A]^{-1} [[zI {-} A](\hat{A}F_0 + B)[zI {-} R] - z^2 B]\}\ , 
\end{align*} 
where $\/F[z]\/$ is proper.  

If $\/x_{-1}\/$, $\/\hat{x}_{1,t}\/$, and $\/u_{-1}\/$ are all zero, 
and $\/\hat{x}_{1,t} {=} \sum_{\tau = 0}^t {F}_{t - \tau} u_\tau\/$, 
then the resulting full model~(\ref{stateqn}-\ref{forecasts}) satisfies 
\begin{align*}
x_t = \sum_{\tau = 0}^t {G}_{t - \tau} u_\tau\ , 
\end{align*} 
where 
\begin{align*} 
G_t & = {\cal Z}^{-1}\{G[z]\} = {\cal Z}^{-1}\{[z^2 \hat{A} - z I + A]^{-1} z [ \hat{A} (\hat{A}F_0 + B)[zI - R] -  B]\} \ .  
\end{align*} 
The forecast error realized at time $\/t \geq 0\/$ is  $\/x_{t} - \hat{x}_{1,t-1} = (\hat{A} {F}_0 + B) w_{t}\/$.  \qed
\end{corollary} 

If $\/{F}[z]\/$ and $\/{G}[z]\/$ are proper, they constitute what are commonly called the {\em transfer matrices\/}, or 
matrices of {\em transfer functions\/}, of a forecasting mechanism and of the corresponding full model, respectively.  
When transfer matrices are proper, they have state-space realizations.  
Such realizations are not unique, and their choice may depend on the details of the structure of the rational matrices.   
But realizations can be found easily with the aid of numerical routines such as the command {\tt tf2ss\/} of the MATLAB Control Systems Toolbox, 
or of Scilab~\citep{scilab}.  
See section~\ref{example} for concrete, numerical examples.  

The especially alert reader will have noticed that the formulas for $\/\tilde{G}[z]\/$ and $\/G[z]\/$ implicitly assume 
the exact cancellation of the term $\/[I - A z^{-1}]^{-1}\/$, arising from the model equation~(\ref{stateqn}) via equation~(\ref{fz}), 
by the term $\/[I - Az^{-1}]\/$ arising in $\/\tilde{F}[z]\/$.  
(To see the cancellation, substitute the solution for $\/\tilde{F}[z]\/$ into equation~(\ref{fz}).)  
However, the implementation of $\/\tilde{F}[z]\/$ given in section~\ref{wellposed} ensures that the latter 
term also arises from the model equation itself, via feedback, in which case the cancellation is a matter of algebra, 
and does not demand infinite precision in the aggregate actions of the public.  

A related concern is that the overall model may be very sensitive to the value of $\/\hat{A}\/$:  
mathematically, variations in that value could give rise to singular perturbations of the nominal 
matrix polynomial $\/[z^2 \hat{A} - zI + A]\/$.  
A thorough treatment of the latter robustness issue is beyond the scope of this paper;  
but an element of a resolution might be to treat $\/\hat{A}\/$ solely as a parameter of the forecasting mechanism, 
whose output would then be the product $\/\hat{A} \hat{x}_{1,t}\/$ 
-- the `economic effect' of the forecast $\/\hat{x}_{1,t}\/$, as opposed to the forecast itself.  

\subsection{Zero-input response}  
\label{ziresp}

To characterize model-consistent forecasting mechanisms fully, the effect of nonzero initial considerations must be considered.  
That of the driving terms having been analyzed, it suffices now (by linearity) to consider the case where the initial conditions may have nonzero values, but all of the driving variables $\/w_t\/$ vanish.  
In system-theoretic terms, the corresponding solution of the model is called its ``zero-input response.''  

By the definition of forecasting mechanisms, in this case, $\/\hat{x}_{1,t} = \overline{x}_{t+1}\/$, where $\/\overline{x}_{t+1}\/$ depends linearly on the initial conditions, but not at all on the $\/w_t\/$.  
For any $\/t \geq 0\/$, model-consistency therefore requires that 
\begin{align*}
x_{t+1} = E_t(\hat{x}_{1,t}) = \hat{x}_{1,t}\ .  
\end{align*} 
In other words, $\/\hat{x}_{1,t}\/$ must be an exact forecast of $\/x_{t+1}\/$.  

The sequence $\/x_t\/$ must therefore be the solution of the following ``perfect-foresight'' model, which captures the case where the driving variables $\/w_t\/$ are zero-valued for all $\/t \geq 0\/$.  
\begin{align*}
\numberthis \label{xxhatpf}
& x_t = A x_{t-1} + \hat{A} \hat{x}_{1,t} + B R^{t+1} u_{-1}\ , \\
& \hat{x}_{1,t-1} = x_{t}\ .  
\numberthis \label{xhatxpf}
\end{align*} 
Substituting for $\/\hat{x}_{1,t}\/$ in the first equation, 
\begin{align*}
x_t = A x_{t-1} + \hat{A}{x}_{t+1} + B R^{t+1} u_{-1}\ ,\ \forall t \geq 0\ .  
\numberthis \label{zixeqn} 
\end{align*} 
Taking unilateral z-transforms and solving for $\/X[z]\/$, 
\begin{align*}
\overline{X}[z] := X[z] = [z^2 \hat{A} - z I + A]^{-1} [z^2 \hat{A} \hat{x}_{1,-1} - z A x_{-1} - z B R [I - R z^{-1}]^{-1} u_{-1}]\ ;  
\numberthis \label{xbarz}
\end{align*} 
Here, the second equation of the perfect-foresight model has been used to replace $\/x_0\/$ with the initial condition $\/\hat{x}_{1,-1}\/$.  

\begin{proposition}
\label{ziprop}
Suppose that 
the model~(\ref{stateqn},\ref{inputeqn}) 
is regular.  
Let 
\begin{align*}
\overline{x}_t & :=  {\cal Z}^{-1}\{\overline{X}[z]\} \\ 
      & = {\cal Z}^{-1}\left \{[z^2 \hat{A} - z I + A]^{-1} [z^2 \hat{A} \hat{x}_{1,-1} - z A x_{-1} - z B R [I - R z^{-1}]^{-1} u_{-1}] \right \}\ .  
\end{align*} 
Then, a perfect-foresight solution exists if and only if $\/\overline{X}[z] - \hat{x}_{1,-1}\/$ is strictly proper.  
In that case, the unique such solution has $\/x_t = \overline{x}_t\/$ and $\/\hat{x}_{1,t-1} = \overline{x}_{t+1}\/$.  
\end{proposition}

\begin{proof}
By the above discussion, any solution $\/x_t\/$ of~(\ref{zixeqn}) must have a unilateral z-transform 
of the form of $\/\overline{X}[z]\/$.  
As a unilateral z-transform, $\/\overline{X}[z]\/$ must be proper.  
Moreover, because $\/x_0\/$ must equal $\/\hat{x}_{1,-1}\/$, 
the matrix $\/\overline{X}[z] - \hat{x}_{1,-1}\/$ must be strictly proper.  
This establishes the necessary condition for existence of a perfect-foresight solution.  

Now, suppose that $\/\overline{X}[z] - \hat{x}_{1,-1}\/$ is strictly proper.  
That implies that $\/\overline{X}[z]\/$ is proper, and that $\/\hat{x}_{1,-1}\/$ is the value of the inverse transform $\/\overline{x}_t\/$ at $\/t=0\/$.  
Rearrange the expression for $\/\overline{X}[z]\/$:   
\begin{align*}
\overline{X}[z] & =  A z^{-1}[\overline{X}[z] + z x_{-1}] + \hat{A} z [\overline{X}[z] - \hat{x}_{1,-1}]  + B R [I - R z^{-1}]^{-1} u_{-1}\ . 
\end{align*} 
Transforming to the time domain, 
that inverse transform $\/\overline{x}_t\/$ is seen to satisfy~(\ref{zixeqn}).   
The uniqueness of this solution follows from that of the transform $\/\overline{X}[z]\/$.  
If, in addition, $\/\hat{x}_{1,t-1} = \overline{x}_{t}\/$, for all $\/t \geq 0\/$, then~(\ref{xxhatpf},\ref{xhatxpf}) is satisfied.  
The full model~(\ref{stateqn}--\ref{forecasts}) 
then satisfies $\/x_t = \overline{x}_t\/$, if $\/w_t = 0\/$,  $\/\forall t \geq 0\/$.  
\qed
\end{proof}
As a consequence of Proposition~\ref{ziprop}, the initial conditions will be said to be {\em consistent\/} if 
$\/\overline{X}[z] - \hat{x}_{1,-1}\/$ is strictly proper; 
and in that case, a model-consistent forecasting mechanism must have $\/\overline{x}_t = {\cal Z}^{-1}\{\overline{X}[z]\}\/$.\footnote{Like the state-space realization of proper rational matrices $\/F[z]\/$ and $\/G[z]\/$, 
the inversion of the proper matrix $\/\overline{X}[z]\/$ can be carried out with the use of standard software tools; 
it amounts to computing the impulse response of a system with a given proper transfer matrix.}    

By the last part of the proposition, $\/\overline{x}_t\/$ can be said, in system-theoretic terminology, to represent 
the {\em zero-input response\/} of the full model~(\ref{stateqn}-\ref{forecasts}), under model-consistent 
forecasts.  
It shows how the system responds to nonzero initial conditions, when any driving terms vanish for $\/t \geq 0\/$.  

However, the nature of the derivation should be borne in mind in potential applications that might otherwise 
exceed the limitations of the results.  
If the model has been evolving through negative time instants, and the model parameters that hold for 
$\/t \geq 0\/$ differ from those in effect for $\/t = -1\/$, it is plausible to suppose that 
realized values of $\/x_{-1}\/$ and $\/u_{-1}\/$ are valid initial conditions for the zero-input response for $\/t \geq 0\/$.   
But it is less clear what are the implications of realizations of the forecast $\/\hat{x}_{1,-1}\/$.  
``Perfect foresight'' need not apply under an unforeseen change in model parameters.    

Like the formulas derived in the previous section, the above solution for $\/\overline{X}[z]\/$ implicitly assumes that a term $\/[I - A z^{-1}]^{-1}\/$ 
arising from the model is exactly cancelled by a term $\/[I - A z^{-1}]^{-1}\/$ derived from the forecast mechanism.  
But this concern can be addressed by realizing the forecast mechanism using feedback from the model equation, as in section~\ref{wellposed}.  

\subsection{Total response} 
\label{totresp}

The total response of the full system is the sum of its zero-state and zero-input responses.  
Indeed, fix any model-consistent forecasting mechanism; 
then $\/x_t\/$ and $\/\hat{x}_{1,t}\/$ must be linear functions of the random variables 
$\/w_\tau\/$, $\/0 \leq \tau \leq t\/$, and of the initial conditions $\/x_{-1}\/$, $\/\hat{x}_{1,-1}\/$, and $\/u_{-1}\/$.  
It follows by linearity that they must then be obtained by summing the separate respective responses 
to the driving variables and to the initial conditions -- namely, the zero-state and zero-input responses.  

\begin{theorem}
\label{exun}
Suppose that the model~(\ref{stateqn},\ref{inputeqn}) is regular.  
For any $\/\hat{A} F_0 \in \mathbbm{R}^{n \times m}\/$, define 
\begin{align*}
{F}_t & = {\cal Z}^{-1}\{F[z]\} =  {\cal Z}^{-1}\{[z^2 \hat{A} - z I + A]^{-1} [[zI {-} A](\hat{A}F_0 + B)[zI {-} R] - z^2 B]\}\ , \\
{G}_t  & = {\cal Z}^{-1}\{ [z^2 \hat{A} - z I + A]^{-1} z[ \hat{A} (\hat{A}F_0 + B)[zI {-} R] -  B]\} \ ,\ \text{ and } \\
\overline{x}_t  & =  {\cal Z}^{-1}\{\overline{X}[z]\} \\ 
       & = {\cal Z}^{-1}\left \{[z^2 \hat{A} - z I + A]^{-1} [z^2 \hat{A} \hat{x}_{1,-1} - z A x_{-1} - z B R [I {-} R z^{-1}]^{-1} u_{-1}] \right \}\ .  
\end{align*}  

Suppose that the initial conditions are consistent ($\/\overline{X}[z] - \hat{x}_{1,-1}\/$ is strictly proper).  
Then there exists a model-consistent forecasting mechanism~(\ref{forecasts}) for (\ref{stateqn},\ref{inputeqn}) 
if and only if $\/F[z]\/$ is proper.   

In that case, in terms of the above inverse transforms, 
the unique model-consistent forecasting mechanism~(\ref{forecasts}) is 
\begin{align*}
\hat{x}_{1,t}  = \sum_{\tau = 0}^t {F}_{t-\tau} \tilde{u}_\tau + \overline{x}_{t+1} \ ,\ \forall t \geq 0\ ; 
\numberthis \label{thmforecast}
\end{align*} 
and the resulting full model~(\ref{stateqn}--\ref{forecasts}) satisfies 
\begin{align*}
x_t  = \sum_{\tau = 0}^t {G}_{t-\tau} \tilde{u}_\tau + \overline{x}_t \ ,\ \forall t \geq 0.  
\numberthis \label{thmtotresponse} 
\end{align*} 
The forecast error realized at time $\/t \geq 0\/$ is 
\begin{align*}
(\hat{A} F_0 + B) w_t = (\hat{A} F_0 + B)(u_t - R u_{t-1})\ .
\end{align*}
\end{theorem}

\begin{proof}
If the initial conditions are consistent, then the necessary form of a model-consistent forecasting mechanism, 
and that of the unique solution of the full model resulting 
from a forecasting mechanism of that form, follow from Corollary~\ref{zscor} and Proposition~\ref{ziprop}, 
by the linearity of the model and of conditional expectations.  
The terms derived from the zero-input solution have no effect on the forecasting errors,  
so (again by linearity) the overall forecasting error realized at time $\/t\/$ is $\/(\hat{A} F_0 + B) w_t\/$, as derived in section~\ref{zsresp}.  
It follows that the specified forecasting mechanism~(\ref{thmforecast}) is indeed model-consistent. \qed 
\end{proof}

It should be emphasized that, assuming that a suitable value of $\/\hat{A} F_0\/$, 
representing the immediate economic effects of shocks on endogenous variables via expectations, 
can be specified, 
the results of this section resolve the nonuniqueness of rational expectations:  
as long as a model-consistent forecasting mechanism exists, 
$\/\hat{A} F_0\/$ determines $\/G_0\/$, 
and ``rationality'' then determines $\/F_t\/$ and $\/G_t\/$ for all $\/t > 0\/$.  

In the next section, a simple assumption is introduced that ensures existence for any value of $\/\hat{A} F_0\/$.  

\subsection{Well-posedness, existence, and feedback}  
\label{wellposed}

A regular model~(\ref{stateqn},\ref{inputeqn}) will be called {\em well-posed\/} if the 
inverse of the ``characteristic matrix'' $\/[-z \hat{A} + I - A z^{-1}]\/$ is proper -- or equivalently, 
if $\/[z^2 \hat{A} - z I + A]^{-1}\/$ is strictly proper.  
For example, this is so whenever $\/\hat{A}\/$ is nonsingular, but not when $\/\hat{A}\/$ is nilpotent.  

Well-posedness admits a simple sufficient condition for consistency of the initial conditions.  
The initial conditions will be called {\em weakly consistent\/} if 
\begin{align*}
\hat{x}_{1,-1} - A x_{-1} - B R u_{-1} \in \text{ Im}\ \hat{A}\ . 
\end{align*} 
Indeed, this condition is plainly necessary if a solution of~(\ref{stateqn},\ref{inputeqn}) is to exist when 
$\/w_t = 0\/$ and the prediction $\/\hat{x}_{1,-1}\/$ is exact.  
Together, well-posedness and weak consistency imply the existence of a model-consistent forecasting 
mechanism for any possible value of $\/\hat{A} F_0\/$.  
In economic terms, well-posedness obviates any assumption that economic actors 
use their presumed aggregate knowledge of the model to ``choose'' 
parameter values $\/\hat{A} F_0\/$ for which solutions exist.    

The theoretical definition of forecasting mechanisms is an unrealistic one, in the sense that, 
incorporating no feedback, a direct implementation would be completely lacking in robustness.  
But it turns out that well-posedness also allows for the realization of model-consistent forecasting 
mechanisms in the form of feedforward/feedback interconnections with the rest of the model, as 
represented by~(\ref{stateqn},\ref{inputeqn}).  
(The condition implies the ``well-posedness'' of that interconnection, 
in the specific sense in which that term is applied to feedback systems.)  
The use of feedback resolves key robustness issues with respect to the model parameter $\/A\/$, but 
sensitivity to the parameter $\/\hat{A}\/$ remains an issue (discussed at the end of section~\ref{zsresp}).  

For brevity, proofs for this section are relegated to appendix~\ref{wellposedproofs}.  
The results themselves are summarized in the following:  

\begin{theorem}
\label{feedbackthm} 
If the model~(\ref{stateqn},\ref{inputeqn}) is regular and well-posed, 
and the initial conditions are weakly consistent, 
then for every possible value of the product  
$\/\hat{A} F_0 \in \mathbbm{R}^{n \times m}\/$, 
there exists a unique model-consistent forecasting mechanism.  
That forecasting mechanism can be realized by the following feedforward/feedback law:  
\begin{align*}
\hat{x}_{1,t} & = \sum_{\tau = 0}^{t} \varPhi_{t-\tau} [A x_\tau + B R u_\tau] - \sum_{\tau=0}^{t} \Psi_{t-\tau} \hat{A} F_0 w_\tau - \Psi_t (\hat{x}_{1,-1} - A x_{-1} - B R u_{-1})\/ .
\end{align*} 
Here, 
$\/\varPhi_t := {\cal Z}^{-1}\{ [I - z \hat{A}]^{-1}\}\/$ and 
$\/\Psi_t := {\cal Z}^{-1}\{ [I - z \hat{A}]^{-1} z \hat{A} \hat{A}^g \}\/$, where $\/\hat{A}^g\/$ is a  
generalized inverse of $\/\hat{A}\/$  (s.t.~$\/\hat{A} \hat{A}^g \hat{A} = \hat{A}\/$).  
\qed 
\end{theorem}

\section{Conventional determinacy of a New Keynesian model}  
\label{example}

In this section it is shown -- purely for purposes of comparison -- 
how the general solution of the previous section lends itself to the reproduction of conventional results.  

Consider for example the loglinearized New Keynesian DSGE model of \citet{RePEc:aea:aecrev:v:94:y:2004:i:1:p:190-217}:  
\begin{align*}
y_t  & = \hat{y}_{1,t} - \tau (r_t - \hat{\pi}_{1,t}) + g_t \numberthis \label{newis}\\
\pi_t & = \beta \hat{\pi}_{1,t} + \kappa (y_t - z_t) \numberthis \label{newpc}\\
r_t  & = \rho_r r_{t-1} + (1 - \rho_r) (\psi_1 \pi_t + \psi_2 [y_t - z_t]) + \epsilon_{r,t} \numberthis \label{newtr}\\
g_t & = \rho_s g_{t-1} + \epsilon_{g,t} \numberthis \label{outputshock}\\
z_t & = \rho_z z_{t-1} + \epsilon_{z,t} \numberthis \label{natoutput}
\end{align*} 
The scalar variables $\/y_t\/$, $\/\pi_t\/$,  and $\/r_t\/$ respectively represent output, inflation, and the nominal interest rate, expressed as percentage deviations from a trend path or a steady state; $\/g_t\/$ and $\/z_t\/$ represent the effects of exogenous shifts on the first two equations.  
In accordance with our usual notation, $\/\hat{y}_{1,t}\/$ and $\/\hat{\pi}_{1,t}\/$ represent forecasts of $\/y_{t+1}\/$ and $\/\pi_{t+1}\/$, dated at time $\/t\/$.  
We shall write $\/w_t = [\epsilon_{g,t} \ \epsilon_{z,t} \ \epsilon_{r,t}]^\prime\/$, and consider distinct values of the resulting vector-valued sequence to be independent, identically distributed, and zero-mean.  

The scalar coefficients are as follows:  $\/\tau\/$ represents intertemporal substitution elasticity, $\/\beta\/$ is the households' discount factor, $\/\kappa\/$ is the slope of the expectational Phillips curve; the third equation describes the monetary authority's behavior, $\/\psi_1\/$ and $\/\psi_2\/$ being `Taylor-rule' coefficients.  

Letting $\/x_t = [y_t \ \pi_t \ r_t]^\prime\/$, $\/u_t = [g_t \ z_t \ \epsilon_{r,t}]^\prime\/$, and $\/w_t = [\epsilon_{g,t} \ \epsilon_{z,t} \ \epsilon_{r,t} ]^\prime\/$, it is an easy matter to put the equations into the form~(\ref{stateqn},\ref{inputeqn}).  
Setting the coefficients equal to the mean values given in Table 1 of \citep{RePEc:aea:aecrev:v:94:y:2004:i:1:p:190-217}, with $\/\beta = 0.99\/$, according to \citep{RePEc:eee:dyncon:v:28:y:2003:i:2:p:273-285}, one finds  
\begin{align*} 
x_t  = & 
\begin{bmatrix} 
    0   &  \ \ 0 & -0.2083333 \\
    0   &  \ \ 0 & -0.1041667 \\
    0   &  \ \ 0 &  \ \ 0.4166667
\end{bmatrix} 
x_{t-1} 
+ 
\begin{bmatrix}
      0.8333333  &  0.1897917 & \ 0  \\
      0.4166667  &  1.0848958 & \ 0  \\
      0.3333333  &  0.6204167 & \ 0 
\end{bmatrix}
\hat{x}_{1,t} \\
& \hspace*{7em} +
\begin{bmatrix}
   0.8333333 & \ \ 0.1666667 & -0.4166667  \\
   0.4166667 & -0.4166667 & -0.2083333 \\
   0.3333333 & \ \ 0.3333333 & \ \ 0.8333333    
\end{bmatrix} 
u_t \ , 
\numberthis \label{NewKstate} \\
u_t  = & \begin{bmatrix}
               0.7 & 0 & 0  \\
               0  & 0.7 & 0 \\
               0 & 0 & 0    
             \end{bmatrix} 
u_{t-1}
+ w_t  \ .
\numberthis \label{NewKinput} 
\end{align*} 
(the respective matrix coefficients being the values of $\/A\/$, $\/\hat{A}\/$, $\/B\/$, and $\/R\/$).  
This model satisfies well-posedness, because all nine entries of $\/[z^2 \hat{A} - zI + A]^{-1}\/$ are strictly proper, 
so a model-consistent forecasting mechanism exists for every possible value of $\/\hat{A} F_0\/$.  

According to the chosen values of the coefficients, $\/\psi_1\/$, the coefficient of inflation in the interest-rate policy `Taylor rule,' has the value $\/1.10\/$.  
If $\/\psi_1 > 1\/$, then when inflation rises, the monetary authority policy raises the interest rate by a greater percentage (all other things being equal):  
such a policy is said to be {\em active\/}; if $\/\psi_1 < 1\/$, policy is called {\em passive\/}.  
With $\/\psi_1 = 0.90\/$, for example, the general solution has only one unstable eigenvalue, while for $\/\psi_1 = 1.10\/$, it has two.\footnote{Throughout the paper, computations were performed in Scilab (version 5.5.2, for Mac OS X).}    
In the conventional approach to rational expectations, the model is found to be ``indeterminate'' in the passive case, but ``determinate'' in the active case:  
only when there are two unstable eigenvalues does their suppression consume sufficient degrees of freedom to yield a unique solution.  
Specifically, when $\/\psi_1 = 1.10\/$, the general solution has unstable eigenvalues at $\/z = 1.4461829\/$ and $\/z = 1.0446352\/$.  
The corresponding left eigenvectors of the denominator polynomial are respectively,\footnote{The results of these and related calculations are displayed with seven or eight significant digits, as a reminder that the calculations really require infinite precision.}   
\begin{align*}
\begin{bmatrix}
-0.5818587  & 0.6738827 & -0.4553268 
\end{bmatrix} 
\& 
\begin{bmatrix}
-0.0473748 & 0.6928388 & 0.7195346
\end{bmatrix} \ .  
\end{align*}
On the other hand, if $\/\psi_1 = 0.9\/$, and policy is ``passive,'' then the model has only one unstable eigenvalue.  

Again, the imposition of stability is carried out here strictly for the purpose of comparison.  
Following \citet{RePEc:aea:aecrev:v:94:y:2004:i:1:p:190-217}, assume that all initial conditions are zero-valued, and therefore focus on the zero-state response.  
Because the matrix $\/R\/$ is stable, it suffices to consider the matrix polynomial $\/G[z]\/$.  
The free parameter $\/\hat{A} F_0\/$ appears only in the numerator of the matrix-fraction description of $\/G[z]\/$, so 
the only way its value can be chosen so as to stabilize an otherwise unstable model is by arranging for unstable `poles' of $\/G[z]\/$ to be 
canceled by `zeros.'  
In the multivariable case, this means that the numerator matrix polynomial of $\/G[z]\/$ must have the same unstable eigenvalues, 
with the same respective left eigenvectors, as the denominator matrix polynomial.  

Note that $\/\lambda_i \in \mathbb{C}\/$ is an eigenvalue of the numerator matrix polynomial of $\/G[z]\/$, with left eigenvector $\/c_i\/$, for both $\/i = 1\/$ and $\/i = 2\/$, if and only if 
\begin{align*}
\begin{bmatrix} c_1 \\ c_2 \end{bmatrix} \hat{A} (\hat{A} F_0 + B)
& = 
\begin{bmatrix}
c_1 B (\lambda_1 I - R)^{-1} \\
c_2 B (\lambda_2 I - R)^{-1}
\end{bmatrix}
\end{align*} 
Each row of the above equation represents three equations, each in a distinct pair of unknowns (the entries from the first two rows of a distinct column of $\/\hat{A} F_0 + B\/$).  
So a single pole-zero cancellation does not determine a unique solution, but the two simultaneous cancellations do.  
The first two columns of 
\begin{align*}
\begin{bmatrix} c_1 \\ c_2 \end{bmatrix} \hat{A} 
\end{align*} 
are linearly independent, yielding a unique solution for the first two rows of $\/\hat{A} F_0 + B\/$; approximately, 
\begin{align*}
\begin{bmatrix}
1.6999275 & \ 0.4900217 &  -0.6182074 \\
1.85166     & -0.5554980 & -0.4620143 
\end{bmatrix} \ .  
\end{align*}
This in turn yields the first two rows of $\/F_0\/$, 
\begin{align*}
\begin{bmatrix}
 0.8094723    &  \quad 0.4571583  & - 0.2066718  \\
    1.0118144  & - 0.3035443  & - 0.1544551 
\end{bmatrix} \ .
\end{align*}
The values of the first two rows of $\/F_0\/$ determine a unique value for $\/\hat{A} F_0\/$ -- approximately, 
\begin{align*}
\begin{bmatrix}
    0.8665942    & \quad 0.3233551   &  - 0.2015408  \\
    1.4349934    & - 0.1388313  & - 0.2536809  \\
    0.8975706    & - 0.0359379  & - 0.1647171  
\end{bmatrix} \ .  
\end{align*}
So the requirement of dynamical stability of the model determines a unique value of $\/\hat{A} F_0\/$, and therefore, 
by Theorem~\ref{exun}, a unique model-consistent forecasting mechanism.  

The resulting matrix $\/G[z]\/$ is the following:  
\begin{align*}
(z - 0.334)^{-1} 
\begin{bmatrix}
1.700 z - 0.949  &  \  \ \ \ 0.490z  - 0.0497  &   - 0.618 z  \\
1.852 z - 0.903     &  - 0.555 z  +  0.271 &   - 0.462 z   \\
1.231 z  \, \ \ \ \ \ \ \ \  \ \  &  - 0.369 z \, \ \ \ \ \ \ \ \ \ \ &    \ \ \ 0.669 z 
\end{bmatrix}
\end{align*} 
It can be realized in the form of the following state-space model:\footnote{In this instance, by means of the Scilab command {\tt tf2ss}.}
\begin{align*}
\zeta_{t+1} & = 0.3343081\ \zeta_t + \begin{bmatrix}   0.8815320  & - 0.2644596    & 0.4788405 \end{bmatrix} u_t \ , \\
x_t & = \begin{bmatrix}  - 0.4316088\\  - 0.3225607 \\ 0.4668023 \end{bmatrix} \zeta_t + 
\begin{bmatrix}  
 1.6999275    & \quad 0.4900217  & - 0.6182074  \\
    1.85166     & - 0.5554980          & - 0.4620143  \\
    1.230904   & - 0.3692712          & \quad 0.6686162 
\end{bmatrix} u_t\ .  
\end{align*} 
Whereas a minimal state-space realization of $\/G[z]\/$ is normally third-order (see section~\ref{leastsquares} for an example), 
at this isolated point in the parameter space of $\/\hat{A}F_0\/$ it is first-order.  
This confirms the two pole-zero cancellations, and reproduces the results of the conventional approach to determinacy.    

In order for the two unstable eigenvalues to be suppressed in this manner, the initial values of the impulse responses of the 
forecasts of output and inflation to the three shocks must equal, exactly, the values given (approximately) by the 
respective entries of the first two rows of $\/F_0\/$ given above:  unless the aggregate forecasts of the public respond to the respective shocks 
in this manner, with infinite precision, the model will be unstable.\footnote{In double-precision floating-point, it takes only a multiplicative perturbation 
of $\/F_0\/$ on the order of $\/1 \pm10^{-9}\/$ to spoil the pole-zero cancellation.}  
The complete lack of robustness of this cancellation makes it untenable that it models any real-world phenomenon.  

Moreover, the next section shows that 
the effect of this suppression is precisely to eliminate the eigenvalues (and the dynamics) that 
arise from expectations.  

\section{Rational expectations and stability} 
\label{stabilitysubs}

The relationship between expectations and stability is a longstanding concern~\citep{AN:exp_stability}.  
This section presents an asymptotic analysis of the effects of expectation terms on dynamics.  

If a `small' scalar multiplicative weight $\/\epsilon\/$ is attached to the matrix coefficient $\/\hat{A}\/$ 
of the forecast term, then unless that matrix is nilpotent,  the denominator matrix $\/[z^2 \epsilon \hat{A} - zI + A]\/$ 
of the overall model is a singular perturbation of that of the lower-order model with $\/\epsilon = 0\/$:  
the degree of the denominator polynomial is a fixed integer greater than $\/n\/$ for all $\/\epsilon > 0\/$;  
but for $\/\epsilon = 0\/$ it is equal to $\/n\/$.  

It turns out that the overall model is always unstable when that weight is positive but sufficiently small.  
Indeed, if $\/z\/$ remains bounded as $\/\epsilon\/$ tends to zero, then the denominator polynomial tends to $\/- [zI - A]\/$, 
so $\/n\/$ of the eigenvalues tend toward those of the matrix $\/A\/$.  
But if $\/\hat{A}\/$ has nonzero eigenvalues, then the denominator polynomial has more than $\/n\/$ finite eigenvalues:   
in order for the degree to drop at $\/\epsilon = 0\/$, some of those finite eigenvalues must `escape' to infinity.  

Indeed, suppose that $\/\hat{A}\/$ has $\/m > 0\/$ nonzero eigenvalues.  
To capture the behavior of eigenvalues that vary like $\/1/\epsilon\/$ 
as $\/\epsilon\/$ tends to zero, 
perform the change of variable $\/z = \lambda/\epsilon\/$.  
The matrix polynomial $\/[z^2 \epsilon \hat{A} - z I + A]\/$ becomes 
\begin{align*}
[\frac{\lambda^2}{\epsilon} \hat{A} - \frac{\lambda}{\epsilon} I + A] & = \epsilon^{-1}[\lambda^2 \hat{A} - \lambda I + \epsilon A] \ .  
\end{align*}
As $\/\epsilon\/$ tends to zero, $\/m\/$ eigenvalues of the matrix polynomial in $\/\lambda\/$ 
approach the reciprocals of the nonzero eigenvalues of $\/\hat{A}\/$.  
Therefore, the moduli of $\/m\/$ eigenvalues of the matrix polynomial in $\/z = \lambda/\epsilon\/$ tend to infinity.  

On the other hand, if $\/\hat{A}\/$ is nonsingular,\footnote{Eigenvalues at infinity complicate the establishment of upper bounds for the finite eigenvalues \citep{HIGHAM20035}.} 
and the weight applied to expectations is sufficiently {\em large\/}, then the eigenvalues of the matrix polynomial in $\/z\/$ will be stable:  

\begin{corollary}
\label{stability}
\mbox{ }\\
{\bf Small expectation gain:}  In equation~(\ref{stateqn}) above, replace the coefficient matrix $\/\hat{A}\/$ with 
$\/\epsilon \hat{A}\/$, where $\/\epsilon\/$ is a real, positive scalar.  
Suppose that $\/\hat{A}\/$ has some nonzero eigenvalue.  
Then, for sufficiently small $\/\epsilon > 0\/$, the denominator matrix polynomial 
$\/[z^2 \epsilon \hat{A} - z I + A]\/$ has unstable eigenvalues; 
consequently, barring pole-zero cancellations, the full model~(\ref{stateqn}-\ref{forecasts}) is dynamically unstable 
under model-consistent expectations.  

\mbox{ }\\
{\bf Large expectation gain:} On the other hand, suppose that $\/\hat{A}\/$ is nonsingular.  
Then the modulus 
of any eigenvalue of $\/[z^2 \hat{A} - z I + A]\/$ 
is at most 
\begin{align*}
\frac{1 + \sqrt{1 + 4 \lVert \hat{A}^{-1}\rVert^{-1} \lVert A \rVert}}{2 \lVert \hat{A}^{-1} \rVert^{-1}}\ .  
\end{align*} 
where $\/\lVert \ldots \rVert\/$ denotes any subordinate matrix norm.  
\end{corollary}

\begin{proof}
The ``small-gain'' result follows from Corollary 1 of \citet{ABG2004}, incorporated 
into a comprehensive theory in \citep{AGM2014a}.  
The upper bound for the case of nonsingular $\/\hat{A}\/$ is from Lemma 3.1 of \citet{HIGHAM20035}.  \qed 
\end{proof}

The above analysis shows that the unstable pole-zero cancellation required to eliminate unstable eigenvalues of the New Keynesian model 
(of section~\ref{example}) under ``active'' policy 
would have the effect of obliterating the dynamical features particularly associated with expectations terms.  
Indeed, the eigenvalues of $\/[z^2 \epsilon \hat{A} - z I + A]\/$ vary continuously with the parameter $\/\epsilon\/$, 
and letting $\/\epsilon\/$ range from one to zero, numerical computation shows that it is precisely 
the two eigenvalues that are unstable under active policy that tend to infinity as $\/\epsilon\/$ tends to zero (moving to the right along the real axis), 
while the other eigenvalues tend toward those of $\/A\/$.  
Now consider reversing the process, starting with $\/\epsilon = 0\/$ and therefore no expectation terms in the model,  
and gradually increasing $\/\epsilon\/$ to $\/1\/$ to restore the expectations.  
The eigenvalues of $\/A\/$ are $\/0\/$, $\/0\/$, and $\/0.417\/$; 
as $\/\epsilon\/$ increases from $\/0\/$ to one, the two eigenvalues at the origin remain fixed, while the third 
shifts from $\/0.417\/$ to $\/0.334\/$.  
One effect of expectations is therefore to produce a modest shift in the nonzero eigenvalue of $\/A\/$.  
The more important dynamical effect of expectations, then, is to bring into being two additional 
modes corresponding to the finite eigenvalues $\/1.045\/$ and $\/1.04\/$.  
However, these are precisely the modes that are suppressed in order to select a unique model-consistent forecasting mechanism 
under the conventional approach to rational expectations.\footnote{The suppression also cancels ``zero dynamics'' of the model, which can explain behavior such as ``price puzzles'' \citep{TM:nofreelunch}.}   
That traditional approach to rational expectations therefore destroys the very dynamical features that arise from expectations.  

The next section presents one method of determining a unique model-consistent forecasting mechanism without resorting to pole-zero cancellation.  
Indeed, it yields a unique solution, regardless of dynamical stability or instability.  

\section{The least-square-error solution}  
\label{leastsquares}

This section proposes a unique model-consistent solution that, at least superficially, seems a natural extension of rational expectations.  
It applies in all cases (regardless of dynamical stability or instability), and does not generally entail pole-zero cancellation.  
In addition to the requirement that forecast errors be zero-mean, it calls for 
all forecast errors be minimized, in the least-squares sense:  
the key to the approach is that a unique value of $\/\hat{A}F_0\/$ will achieve such minimization.  

Specifically, let $\/e_t\/$ denote $\/(\hat{A} F_0 + B) w_t\/$, the forecast error realized at time $\/t\/$, under model-consistent expectations.  
Then it is assumed that $\/\hat{A} F_0\/$ is such that, given the sequence of the $\/w_t\/$, every 
squared-error term $\/e_t^\prime e_t\/$ is minimized.  

This criterion is not proposed as a panacea.  
On the contrary, it has features that may be undesirable, at least in some contexts.  
Because $\/G_0 = \hat{A} F_0 + B\/$, error minimization amounts to a minimization (in the least-squares sense) 
of the immediate response of the endogenous variables in $\/x_t\/$ to shocks.  
It implies that $\/\hat{A}F_0\/$ is `selected' in such a way that the immediate response of the model to shocks via expectations 
$\/\hat{x}_{1,t}\/$ at least partially cancels its immediate response to shocks via the exogenous variables $\/u_t\/$.  
Partial `blocking' of the effect of shocks by expectation terms is a fixture of the pole-zero cancellation employed in the traditional solution 
of rational-expectations models, but that does not imply that it is realistic.  
If the column span of $\/\hat{A}\/$ contains that of $\/B\/$, then 
$\/\hat{A}F_0\/$ can be chosen so that $\/\hat{A} F_0 + B\/$ is zero, and the immediate effects of shocks are fully blocked:  
the model will then exhibit `perfect foresight,' or `self-fulfilling expectations,'  and the minimum 
squared-error terms will of course all be zero.  
In particular, the above will apply whenever $\/\hat{A}\/$ is nonsingular -- as it must be in the scalar case, for instance.    
Another consequence of the vanishing of $\/\hat{A}F_0 + B\/$ is that $\/G_0\/$, 
the initial value of the impulse response of the 
overall model vanishes, and $\/G[z]\/$ is strictly proper.  

But while this new assumption clearly represents a significant strengthening of the usual rational-expectations hypothesis, it is in a similar spirit; 
and in comparison with the infinite-precision, unstable pole-zero cancellations that are usually assumed in conjunction with rational expectations, 
it is arguably relatively mild.  
It preserves the key feature that expectations are policy-dependent, if, as in the 
example of section~\ref{example}, the policy parameters are reflected in the matrix coefficients of the model.  
But it does not require specific stability properties of the model, 
and it generally avoids the feature that the unique solution has reduced-order dynamics, owing to pole-zero cancellation.  

To define the unique solution, note that each column $\/b\/$ of the matrix $\/B\/$ can be uniquely decomposed into $\/b_\parallel + b_\perp\/$, 
where $\/b_\parallel\/$ is its projection onto the column span of $\/\hat{A}\/$, and $\/b_\perp\/$ is orthogonal to that vector space.  
Let $\/B_\parallel\/$ consist of the projections $\/b_\parallel\/$ of the respective columns of $\/B\/$, and $\/B_\perp\/$ of the 
respective orthogonal components $\/b_\perp\/$.  
Then, for any $\/t \geq 0\/$, 
\begin{align*}
e_t^\prime e_t  & = w_t^\prime (\hat{A} F_0 + B)^\prime (\hat{A} F_0 + B) w_t \\
                         & = w_t^\prime (\hat{A} F_0) + B_\parallel + B_\perp)^\prime (\hat{A} F_0  + B_\parallel + B_\perp) w_t \\
                         & = w_t^\prime (\hat{A} F_0 + B_\parallel)^\prime (\hat{A} F_0 + B_\parallel) w_t + w_t^\prime B_\perp^\prime B_\perp w_t \ .  
\end{align*} 
Because this is a sum of nonnegative quantities, its minimum value over all possible $\/\hat{A}F_0\/$ must be at least 
$\/w_t^\prime B_\perp^\prime B_\perp w_t\/$.  

Now, for any column $\/b\/$ of $\/B\/$, the corresponding column $\/f_0\/$ of $\/F_0\/$ can be chosen (not necessarily uniquely) 
so that $\/\hat{A} f_0 = - b_\parallel\/$.  
Let $\/F_0\/$ be composed of such columns.  
Then $\/\hat{A} F_0 = - B_\parallel\/$, and; for any $\/t \geq 0\/$, 
\begin{align*}
e_t^\prime e_t  & = w_t^\prime B_\perp^\prime B_\perp w_t \ .  
\end{align*} 
This unique value $\/\hat{A} F_0 = - B_\parallel\/$ therefore achieves the minimum squared error, for any given $\/w_t\/$.  

The foregoing discussion establishes the following 
\begin{corollary}
\label{leastsquarederror}
Suppose that the model~(\ref{stateqn},\ref{inputeqn}) is regular.  
Let $\/\hat{A} F_0 = - B_\parallel\/$.  
Suppose that the resulting $\/F[z]\/$ is proper and the initial conditions are consistent.  
Then the forecasting mechanism given by Theorem~\ref{exun} is the unique 
model-consistent forecasting mechanism that gives rise to least-square forecasting errors.  
By Theorem~\ref{feedbackthm}, a weaker sufficient condition is that the model be well-posed and 
the initial conditions weakly consistent.  
\end{corollary} 

To be sure, the approach outlined above requires the public, in its aggregate, 
to calibrate its immediate reaction to shocks so as to meet the requirement 
of least-square forecasting errors.  
But here, the unique solution is not as ill-conditioned or `brittle' as under the method of pole-zero cancellation:  
a small error in the value of $\/\hat{A}F_0\/$ generally means only that forecast errors will be slightly larger than necessary; whereas in the case 
of pole-zero cancellation, the slightest error means that the model will be dynamically unstable.\footnote{In fact, in instances where the least-square-error criterion calls for pole-zero cancellation, a practical means of producing a less pathological solution might be simply to apply a `small,' random perturbation to the value of $\/\hat{A}F_0\/$.}    

\subsubsection*{New Keynesian example} 

For the example of section~\ref{example}, we find that, 
\begin{align*}
- B_\parallel 
& = 
\begin{bmatrix}
  - 0.833  & \, - 0.155  &  \quad 0.322  \\
  - 0.417  & \quad 0.469  & \, - 0.209   \\
  - 0.333  & \quad 0.239  & \, - 0.075 
\end{bmatrix} \ .  
\end{align*} 
Setting $\/\hat{A}F_0 = - B_\parallel\/$, and finding a minimal state-space realization for the resulting matrix $\/F[z]\/$, we arrive at a representation of the unique, least-squared-error, model-consistent forecasting mechanism:  
\begin{align*}
\xi_{t+1} & = 
\begin{bmatrix}
   \quad 1.574  & \, - 0.937  & \, - 0.835  \\
    \, - 0.0942    & \quad 0.978    & \quad 0.285  \\
    \quad 0.271 & \quad 0.109    & \quad 0.273 
\end{bmatrix} 
\xi_t 
+ 
\begin{bmatrix}
        \, - 0.288 & \, - 1.153  &     \quad 0.370  \\
  \quad  1.003 &  \quad 0  &   \, - 0.308           \\
         \quad 0  &   \quad 0 &   \, - 0.671  
\end{bmatrix}
u_t \ , \numberthis \label{fzseqn} \\
\hat{x}_{1,t} & = 
\begin{bmatrix}
  \quad  0.488  & \, - 1.171  & \, - 0.528  \\
  \, - 0.592   & \quad 0.334   & \quad 0.437  \\
  \, - 0.349  & \, - 0.049  & \, - 0.0059  
\end{bmatrix}
\xi_t + 
\begin{bmatrix}
           \, - 1  &  \, - 0.311  & \quad 0.471 \\
   \quad  0  & \quad 0.552    & \, -0.374 \\
     -0.125   &  \quad 0.130   & \quad 0.233 
\end{bmatrix}
u_t \ .  \numberthis \label{fzoeqn}  
\end{align*} 
The second matrix coefficient of the second equation being the initial value of the system's impulse response, it equals $\/F_0\/$;   
multiplying it by $\/\hat{A}\/$ yields the above $\/-B_\parallel\/$.  
The model is well-posed, so this forecasting mechanism could be alternatively be represented as a feedback/feedforward predictor, according to Theorem~\ref{feedbackthm}.  

Coupling the forecasting mechanism to the model equation~(\ref{NewKstate}) 
yields a system that can be represented in the form of the following state-space model 
(obtained as a realization of the corresponding rational matrix $\/G[z]\/$):  
\begin{align*}
\zeta_{t+1} & = 
\begin{bmatrix}
   \quad 1.574 & \, - 0.937  &  \quad 0.835  \\
  \, - 0.094  &  \quad 0.978  & \, - 0.287  \\
  \, - 0.271  & \, - 0.109  &  \quad 0.273  
\end{bmatrix} 
\zeta_t 
+ 
\begin{bmatrix}
  \, - 0.290 & \, - 1.161 &  \quad 0.373 \\  
   \quad 1.011 & \quad 0 &  \, - 0.310  \\
    \quad 0 &  \quad 0 & \quad 0.676  
\end{bmatrix}
u_t  \ , \numberthis \label{gzseqn} \\
x_t & = 
\begin{bmatrix}
   \quad 0.275  & - 0.911  & \quad  0.128  \,  \\
  \, - 0.444  & - 0.127  & \, - 0.366 \,  \\
  \, - 0.169  & - 0.172    & \quad 0.358  \, 
\end{bmatrix}
\zeta_t + 
\begin{bmatrix}
  \quad 0 & \quad  0.0118  & \, - 0.095 \quad \\
  \quad 0 & \quad 0.0522   &\,  - 0.417 \quad \\
  \quad 0 & \, - 0.0948         & \quad 0.759 \quad 
\end{bmatrix}
u_t \ .  \numberthis \label{gzoeqn}
\end{align*} 
The coefficient $\/\hat{A}F_0 + B\/$ that determines the forecast errors is the same as $\/G_0\/$, and therefore equal to 
the second matrix coefficient of equation~(\ref{gzoeqn}).  

In contrast with the traditional solution, this one does not have a reduced dynamical order, 
and the preservation of model dynamics turns out to reveal difficulties with the formulation of the New Keynesian model.  
Indeed, the above state-space equations are at odds with the standard interpretation of key model equations.  
For instance, the New I-S equation~(\ref{newis}) is typically considered to assert that the higher the real interest rate $\/r_t - \hat{\pi}_{1,t}\/$, 
the lower output $\/y_t\/$; and by the same token, the greater $\/g_t\/$, the higher the output $\/y_t\/$.  
But the state-space equations disagree.  
Note that $\/g_t\/$ is the first component of the vector $\/u_t\/$, and 
$\/y_t\/$ is the first component of $\/x_t\/$.  
According to~(\ref{gzoeqn}), $\/g_t\/$ has no immediate effect upon any of the 
components of $\/x_t\/$; and by~(\ref{fzoeqn})  its only effect on the 
forecasts of interest is a negative effect on $\/\hat{y}_{1,t}\/$.  

But it is not only the least-squared-error solution that disagrees with the usual interpretation of the New Keynesian model.  
It is shown in the next section that, barring cancellations of eigenvalues of $\/[z^2 \hat{A} - z I + A]\/$, 
an analysis of the model eigenvalues under any model-consistent forecasting system 
conflicts with the usual interpretations of both the I-S equation and 
the expectational Phillips curve.  

\section{Stabilization of the New Keynesian model}  
\label{stabilization}

By decoupling rational expectations from dynamical stability, the methods of this paper give a clearer picture of the dynamics of macroeconomic models.  
An illustration is provided by the formulation of a stabilizing interest-rate policy for the New Keynesian model.  

If the $\/R\/$ matrix is stable, then in the absence of pole-zero cancellations, the stability of the overall model  
depends only on the eigenvalues of the matrix polynomial $\/[z^2 \hat{A} - z I + A]\/$.  
Consequently, one can play at a toy version of central banking by 
altering the Taylor-rule parameters in such a way that all of that matrix polynomial's eigenvalues become stable.  
For instance, with $\/\psi_1 = 1.10 \/$ and $\/\psi_2 = -1.50\/$, the eigenvalues are 
$\/0.81 \pm 0.045 i\/$ and $\/0.76\/$ (while their moduli can be decreased further by decreasing the value of $\/\rho_r\/$).  
These indicate a stable model, with a response that is only slightly oscillatory.  
Under the conventional approach to rational expectations, these Taylor-rule parameters would not produce a `determinate' model, 
dynamical stability being incompatible with determinacy.  But in general, there is no such difficulty:  
for instance, this stabilization procedure could be carried out in conjunction with the use of the above least-square-error criterion.  

Note that the value of the above Taylor-rule coefficient $\/\psi_2\/$ applied to output is negative, 
meaning that the higher the level of output, the lower is the policy interest rate.  
The reason for this surprising feature lies in the peculiarity of the new I-S equation.  
The equation is derived from microfoundations, and specifically from the solution of optimal planning problems 
for members of different sectors of the economy.  
Households, for example, are assumed to schedule their consumption over all time, according to expected values of 
nonlinear functions of interest rates and other variables.  
Under that scenario, the higher the expected real interest rate at time $\/t\/$, the lower consumption will indeed be at the same 
time -- moreover, the lower consumption will be at earlier time instants as well -- the better to profit from a higher rate.  

Linear model equations are then obtained by ``loglinearizing'' the first-order conditions associated with optimality -- 
a process of informal approximation that typically assumes, for instance, that expected values of nonlinear functions 
can be approximated as expected values of approximations of those functions.  
The new I-S equation is derived by identifying output with consumption in the resulting system of equations.  
Unlike the solution of the consumption planning problem, this equation of course does not imply that the higher the 
present real interest rate, the lower earlier levels of consumption.  
In fact, contrary to the usual interpretation, it is difficult to argue that it implies a 
negative relationship even with the current level of output:  
for any alteration in the current real interest rate there will almost surely be a cognate change in the 
forecast $\/\hat{y}_{1,t}\/$ of the next period's output, so it is not obvious how $\/y_t\/$ will be affected. 
Better to rearrange the equation as follows, 
\begin{align*}
\hat{y}_{1,t} - y_t  & =  \tau (r_t - \hat{\pi}_{1,t}) - g_t \ , \numberthis \label{newisrearr} 
\end{align*}
and see that, with $\/\tau > 0\/$, it means simply that, the larger the real interest rate, 
the larger the expected increase in output over the next period.  

But more can be said in the context of the overall model:  
note that the denominator matrix polynomial $\/[z^2 \hat{A} - z I + A]\/$ is exactly the same 
in the case of the zero-input response (section~\ref{ziresp}) as for the zero-state response~(\ref{zsresp}).  
In order to study model eigenvalues it therefore suffices (in the absence of pole-zero cancellations) to 
consider the zero-input response.  
But if the sequences $\/\epsilon_{r,t}\/$, $\/\epsilon_{g,t}\/$, $\/\epsilon_{z,t}\/$ are identically zero, 
then, according to the analysis of section~\ref{ziresp}, the I-S equation and the expectational Phillips curve effectively become,    
\begin{align*}
y_{t+1} & =  y_t + \tau (r_t - \pi_{t+1}) - g_t \ , \numberthis \label{newisperf} \\
\pi_{t+1} & = \beta^{-1}  \pi_t - \beta^{-1} \kappa (y_t - z_t)\ . \numberthis \label{newpcperf}
\end{align*}
The result is a perfect-foresight version of the New Keynesian model, 
that gives rise to the same eigenvalues as the general version.  
For any given initial conditions, it has a unique 
solution, according to which the values of $\/\pi_{t+1}\/$ and $\/y_{t+1}\/$ are determined by their previous values and by $\/r_{t}\/$, 
$\/g_{t}\/$ and $\/z_{t}\/$.  
It follows that the larger $\/r_t\/$, the {\em larger\/} $\/y_{t+1}\/$.  
So the greater the interest rate, the {\em greater\/} will be output at the next instant.  
Similarly, with $\/\beta\/$ essentially unity and $\/\kappa > 0\/$, 
the expectational Phillips curve~(\ref{newpc}) effectively implies that 
the higher the level of output, the {\em lower\/} the expected rate of increase of inflation over the next period.  
For use with model-consistent expectations, both equations are therefore flawed.  

In determining inflation in the full model, the two inversions essentially cancel each other, 
because the only transmission channel from the interest rate to inflation goes by way of output.  
That explains why the sign of $\/\psi_1\/$ in the above stabilizing policy rule conforms to economic intuition, while that of $\/\psi_2\/$ does not.  
On the other hand, if one crudely `corrects' the two key equations by simply changing the signs of $\/\tau\/$ and $\/\kappa\/$, 
one finds -- as expected -- that in stabilizing policy rules, both $\/\psi_1\/$ and $\/\psi_2\/$ are positive.  
Indeed, for stability, the policy reactions to both inflation and the output gap should be `active':  
with $\/\psi_2 = 1.5\/$, the model is stable provided $\/1.03 \leq \psi_1 \leq 1.49\/$; and with $\/\psi_1 = 1.03\/$, the model is stable provided $\/1.04 \leq \psi_2 \leq 1.50\/$.  
In the toy game of central banking, one might opt for the policy gains $\/\psi_1 = 1.10\/$ and $\/\psi_2 = 1.50\/$:  that would give eigenvalues at $\/0.812 \pm 0.0453 i\/$ and $\/0.763\/$; these are about as far from the unit circle as can be arranged with those two policy parameters, and they produce a response that is only slightly oscillatory.\footnote{The moduli of the eigenvalues could be reduced further by reducing the value of $\/\rho_r\/$ that serves to smooth changes in the policy interest rate.}  

In a recent, constructive critique of DSGE models, \citet{RePEc:oup:oxford:v:34:y:2018:i:1-2:p:43-54.}  
asserts that the new I-S equation (\ref{newis}) and the expectational Phillips curve~(\ref{newpc}) are ``deeply flawed.''  
The results of this section lend support to that statement, but not for all of the reasons that Blanchard cites.  
He describes rational expectations as insufficiently ``inertial.'' 
The fundamental reason for that assertion is likely the suppression of the associated dynamics that is part and parcel of the conventional application of rational expectations.  
However, that suppression is not a general feature of rational expectations, but only of the very particular approach that has held sway for the last several decades.  

\section{Related work} 
\label{related}

In seeking a way forward for DSGE models, \citet{RePEc:oup:oxford:v:34:y:2018:i:1-2:p:43-54.} asks, 
``how can we deviate from rational expectations, while keeping the notion that people and firms care about the future?''  
This paper suggests that the problem is not that rational expectations have been tried and found wanting, 
but rather that they have never been tried (with apologies to G.K.~Chesterton).       

The main reason is that the cause of nonuniqueness was never before understood, 
chiefly because earlier work did not include any explicit parameterization of the forecasting mechanism.  
\citet{RePEc:ecm:emetrp:v:45:y:1977:i:6:p:1377-85} adapted the early method of \cite{Muth:rational} to 
dynamic macroeconomic models by means of a parameterization of the overall model (see subsection~\ref{Taylor}) in the form of a Wold decomposition.  
While Taylor's formulation resembles the present one in some respects, it assumes dynamical stability, an unnecessary limitation.  
More to the point, Taylor did not explicitly parameterize the forecasting mechanism.  
\citet{RePEc:eee:moneco:v:4:y:1978:i:1:p:1-44} recognized that the key to the solution of rational-expectations models was 
to solve for the forecasting mechanism, but neither did he include such an explicit parameterization.  
The same is true of all other previous approaches;  
for recent examples, see \citep{TAN201595} and \citep{al-sadoon_2017}.  
Yet the free parameters that underlie the nonuniqueness of solutions are essentially 
parameters of the forecasting mechanism itself.  

Without insight into the true reason for nonuniqueness, 
and perhaps because of a generally perceived need for a terminal boundary condition,
stability quickly became entrenched as the primary means of attempting to select a unique solution.  
\citet{RePEc:ecm:emetrp:v:45:y:1977:i:6:p:1377-85} and \citet{RePEc:aea:aecrev:v:69:y:1979:i:2:p:114-18} 
considered alternatives, 
but -- invoking the stability constraint himself only shortly afterward -- \citet{RePEc:aea:aecrev:v:71:y:1981:i:1:p:132-43} 
referred to its application as having already become 
``a standard if not entirely convincing practice.''  
It has remained so for decades, even though far from universally applicable, 
and in spite of the recognition that it obliterates model dynamics and alters stability properties \citep{RePEc:aea:aecrev:v:94:y:2004:i:1:p:190-217}.  

Consequently, standard methods of computing rational-expectations solutions 
are needlessly bound up with the question of stability.  
It is generally assumed that an equation of the following form holds:  
\begin{align*}
\Gamma_0 y_t = \Gamma_1 y_{t-1} + \Psi v_t + \Pi \eta_t\ .  
\end{align*} 
The vector $\/y_t\/$ is made up of endogenous variables, including forecasts, 
$\/v_t\/$ is an exogenous random driving variable, 
and $\/\eta_t\/$ is a vector of zero-mean forecasting errors, 
that is determined as part of the solution process.\footnote{See, for example, \citep{RePEc:ecm:emetrp:v:48:y:1980:i:5:p:1305-11,RePEc:cup:etheor:v:13:y:1997:i:06:p:877-888_00,RePEc:eee:dyncon:v:24:y:2000:i:10:p:1405-1423,RePEc:kap:compec:v:20:y:2002:i:1-2:p:1-20,RePEc:aea:aecrev:v:94:y:2004:i:1:p:190-217}.}  
A matrix decomposition is applied, to allow separate consideration of the respective eigenvalues, 
and where unstable eigenvalues are concerned, it is arranged (if possible) for the corresponding terms in $\/v_t\/$ and $\/\eta_t\/$ to 
cancel each other, giving rise to a zero that cancels the unstable eigenvalue.  
If the corresponding calculations determine the forecasting errors uniquely, then there is a unique rational-expectations 
solution -- among those that obey the stability constraint.  
As this article proves, this is merely an arbitrary and indirect means of unknowingly modeling the public's immediate responses to shocks.  
But the approach has become so closely identified with the solution of rational-expectations models that it should be stressed that the link with 
stability is completely unnecessary -- as is the unrealistic, exact cancellation of unstable eigenvalues, and the consequent suppression of dynamics.  

\citet{RePEc:ecm:emetrp:v:45:y:1977:i:6:p:1377-85} considered, as an alternative means of ensuring uniqueness, 
the requirement that the variance of the endogenous variables be minimized.  
But  \citet{BASAR1989591} appears to have been the first to study explicitly the minimization of the variance of forecast errors.  
Ba\c{s}ar considered a univariate model with a single forecast term $\/\hat{x}_{2,t-1}\/$, 
driven by an independent, zero-mean sequence with finite variance, and assumed that the information 
set underlying forecasts consisted either of $\/x_{t-1}\/$, or of a noisy measurement of that scalar variable, 
with all random variables being Gaussian in the latter case.  
He showed that a minimum-variance rational-expectations solution could be computed as the limiting 
case of solutions to finite-horizon problems -- under an assumption that amounts to the realness of the two 
eigenvalues of the overall model.  
Specifically, the finite-horizon problems called for minimization of a discounted sum of squared forecast errors, 
over a finite time interval, and subject to a terminal boundary condition.  
Apart from the restriction on information sets, Ba\c{s}ar made no assumptions as to the form of forecasting mechanisms, 
but his solutions satisfied linear recurrences with constant coefficients.  
General solutions forTaylor's and Ba\c{s}ar's models are given by the results of appendix~\ref{multexp}, together with formulas for the forecast errors.  
In both cases, the space of model-consistent solutions can be characterized by two parameters, $\/\tilde{G}_0\/$ and $\/\tilde{G}_1\/$, of which 
the first is determined (because the models contain no unlagged expectation terms) and the second is completely free (because the models satisfy the appropriate 
generalization of the well-posedness condition).  
Forecast-error variances are positive, and are minimized by setting $\/\tilde{G}_1 = 0\/$  
-- this amounts to nulling the initial value of the impulse response of an unlagged one-step forecast.  

\section{Conclusion}
\label{conclusion}

This paper has shown that the free parameters that underlie the nonuniqueness of solutions of 
linear rational-expectations models are coefficients that capture the immediate reactions of economic 
actors to shocks.  
It has presented a general solution that has a unique instance for every appropriate specification 
of those free parameters.  
If the model is {\em well-posed\/}, and the initial conditions {\em weakly consistent\/}, 
then there exists a (unique) solution for every possible specification of the parameters.  

If, in addition, the requirement of model-consistency is augmented with that of least-square forecast errors, 
then -- for a broad class of models -- unique values of the free parameters are determined, and therefore so is a unique solution.  
This result is offered as one concrete means of ensuring uniqueness.  
It avoids the most serious drawbacks of the prevalent approach, while preserving its advantages, 
but its suitability for a given application should be considered critically in the light of the paper's main results, 
and the corresponding form of the solution that it yields.  

Indeed, the fundamental objective of the paper has been to explain the nature of the nonuniqueness of rational expectations, 
so that arbitrary, ad hoc solutions might be avoided.  
In that connection, the main finding is that model consistency is not as strong an assumption as may have 
been supposed:  it generally implies nothing about the 
economic effects of the most immediate responses of forecasts to shocks.    
To put it another way, there is no fundamental reason for any model of the most  
short-term reactions of economic agents to entail systematic forecasting errors.  
A rough paraphrase suggests itself, in the 
terminology of \citet{Kahneman:thinking}:   
rational expectations exemplify analytical, ``slow'' thinking, and as such do not 
capture the ``fast'' thinking that underlies the most immediate responses to 
economic shocks -- but clearly, any deliberate and reasoned process of expectation 
formation must necessarily take account of the economic effects of more instantaneous, reflexive actions.  
Like most loose paraphrases, the analogy probably should not be pushed too far, 
but it should be underlined that ``fast thinking'' does not necessarily imply on-the-fly 
improvisation, and could include the automatic application of a preconditioned response, 
such as might indeed be represented by the free parameters identified here.  

By dispensing with the cancellation of unstable dynamics, the results of the paper remove substantial obstacles 
to the application and the study of rational expectations.  
They greatly expand the range of models to which rational expectations can be applied, and they avoid the suppression of model dynamics.  
Consequently, they allow the examination of fundamental questions of stability and of stabilization.  
They should also simplify problems that are central to the pertinence of the 
rational expectations hypothesis, such as those of model estimation and the ``learning'' of models by 
economic agents~\citep{RePEc:eee:moneco:v:4:y:1978:i:1:p:1-44,RePEc:red:issued:13-148}.    

More generally, these new results provide a richer picture of the dynamics of macroeconomic models, 
and allow for more reasoned and realistic means ``to attribute to individuals some view of the behavior of the future values 
of variables of concern to them'' \citep{RePEc:eee:crcspp:v:1:y:1976:i::p:19-46}.  

\vfill 
\pagebreak

\bibliography{econrefs}
\bibliographystyle{econ}

\vfill 
\pagebreak 

\appendix

\section{Appendix:  a ``general'' model}  
\label{multexp}

Within the traditional rational-expectations paradigm, 
the model employed in the main body of the paper 
has been applied to cases of arbitrarily many distinct expectation terms  
through an increase of the dimension $\/n\/$ 
\citep{broze_gouriŽroux_szafarz_1995,RePEc:cup:etheor:v:13:y:1997:i:06:p:877-888_00,RePEc:eee:dyncon:v:31:y:2007:i:4:p:1376-1391}.  
The present method extends to allow models with a variety of expectation terms to be treated directly.  
A direct analysis lends insight into the central question of the paper --   
that of nonuniqueness of solutions and the associated free parameters --   
and of course allows direct solution of a broader class of models.  

To outline a generalized approach, take the following model:  
\begin{align}
\label{multstateqn}
& \sum_{i = 0}^h \sum_{j=0}^l A_{ij} \hat{x}_{i,t-j} = B u_t\ ,\ A_{00} = I; \\
\label{multinputeqn}
& u_t = R u_{t-1} + w_t\ .  
\end{align}
Here, the $\/A_{ij} \in \mathbbm{R}^{n \times n}\/$ are constant matrix coefficients, as are 
$\/B \in \mathbbm{R}^{n \times m}\/$ and \linebreak $\/R \in \mathbbm{R}^{m \times m}\/$.  
The sequence $\/w_t\/$ is as in the main body of the paper, 
and the vectors $\/u_t\/$ again represent exogenous inputs driven by the $\/w_t\/$.  

The vector $\/\hat{x}_{0,t} \in \mathbbm{R}^n\/$ is a vector $\/x_t\/$ of endogenous variables, and, for $\/0 < j \leq l\/$, 
$\/\hat{x}_{0,t-j} = x_{t-j}\/$ is a ``lagged'' version of $\/x_t\/$.  
For positive $\/i\/$, $\/0 < i \leq h\/$, 
$\/\hat{x}_{i,t}\/$ represents a forecast, formulated at time $\/t\/$, of the value of $\/x_{t+i}\/$.  
For brevity, this appendix focuses on the novel part of the present approach, the formulation and solution of the zero-state response: 
hence, all initial conditions will be assumed to be zero-valued, and forecasts will depend only 
on the random variables $\/u_\tau\/$, for $\/0 \leq \tau \leq t\/$.  

Formally, a {\em forecasting mechanism\/} takes the form of a system of equations of the form 
\begin{align*}
\hat{x}_{i,t-j} = \sum_{\tau = 0}^{t-j} F_{ij,t-\tau} u_\tau,\ \forall\ 0 < i \leq h\ ,\ 0 \leq j \leq l\ , t \geq j\ , 
\numberthis \label{multforecasts} 
\end{align*} 
where $\/F_{ij,t} \in \mathbbm{R}^{n \times m}\/$.  
The upper limits of the convolution sums reflect the fact that the forecast 
$\/\hat{x}_{i,t-j}\/$ must be based solely on information available at time $\/t-j\/$.  
More could be assumed about the form of the convolution terms, 
but additional structure will make itself evident shortly.  

As before, it will be convenient also to consider forecasting mechanisms driven directly by the $\/w_t\/$, 
\begin{align*}
\hat{x}_{i,t-j} = \sum_{\tau = 0}^{t-j} \tilde{F}_{ij,t-\tau} w_\tau,\ \forall i,j\ 0 < i \leq h\ ,\ 0 \leq j \leq l\,\ t \geq j\ .    
\numberthis \label{multtildeforecasts} 
\end{align*} 

If $\/y\/$ is a square-integrable random variable defined on the common probability space of the $\/w_t\/$, 
then $\/E_t(y)\/$ denotes the expected value of $\/y\/$, 
conditioned on $\/w_\tau\/$, for $\/0 \leq \tau \leq t\/$.  
A forecasting mechanism is {\em model-consistent\/} if the resulting full model 
satisfies  the following for all $\/t \geq j-1\/$, $\/0 \leq i \leq h,\  0 \leq j \leq l\/$, 
\begin{align*} 
        & E_{t-j}(x_{t+i-j} - \hat{x}_{i,t-j}) = 0  \\  
\iff     & \hat{x}_{i, t-j} = E_{t-j}(x_{t+i-j})\ .  
\end{align*} 

It is important to note that the expected values are subject not only to the model equations (\ref{multstateqn},\ref{multinputeqn}), 
but also to the forecasting mechanism.  
In other respects, the above model resembles the ``general model'' of \citet{broze_gouriŽroux_szafarz_1995}.  

The $\/w_t\/$ influence the $\/x_t\/$ causally, not only through the expectations terms but also via the exogenous inputs.  
The $\/x_t\/$ must therefore constitute a convolution of the $\/w_t\/$ with an impulse response, of the following 
form:   
\begin{align*}
x_t = \sum_{\tau = 0}^t \tilde{G}_{t - \tau} w_\tau\ .  
\numberthis \label{multgtconv}
\end{align*} 

For equations relating the $\/\tilde{G}_t\/$ and the $\/\tilde{F}_{ij,t}\/$, invoke model-consistency, for $\/0 < i \leq h\/$, and $\/t \geq 0\/$:  
\begin{align*}
\sum_{\tau = 0}^{t} \tilde{F}_{ij,t - \tau} w_\tau  = \hat{x}_{i,t-j}  = E_{t-j}(x_{t+i-j}) & = E_{t-j}\left ( \sum_{\tau = 0}^{t+i-j} \tilde{G}_{t+i-j-\tau} w_\tau  \right ) \\
& = \sum_{\tau = 0}^{t-j} \tilde{G}_{t+i-j-\tau} w_\tau\ .  
\end{align*} 
Taking expected values of the left- and right-hand sides, conditioned on $\/w_0\/$, 
\begin{align*}
\tilde{F}_{ij,t} w_0 = \mathbbm{1}_{t-j} \tilde{G}_{t+i-j} w_0\ ,\ \forall t \geq 0\ .     
\end{align*} 
Here, $\/\mathbbm{1}_t\/$ denotes the unit-step function, which vanishes for negative arguments but otherwise is unity.  
Because $\/w_0\/$ is arbitrary,    
\begin{align*}
\tilde{F}_{ij,t} = \mathbbm{1}_{t-j} \tilde{G}_{t+i-j}\ ,\  0 < i \leq h\ ,\ t \geq 0\ .  
\end{align*} 
Given the above form of the $\/\tilde{F}_{ij,t}\/$, drop the index $\/j\/$ and write instead 
\begin{align*}
\tilde{F}_{i,t} = \mathbbm{1}_{t} \tilde{G}_{t+i}\ ,\  0 < i \leq h\ ,\ t \geq 0\ ,  
\numberthis \label{fig}
\end{align*} 
and 
\begin{align*}
\hat{x}_{i,t-j} 
= \sum_{\tau = 0}^{t-j} \tilde{F}_{i,t-j-\tau} w_\tau\ .  
\numberthis \label{multftconv}
\end{align*} 

After substitution of the above convolutions into~(\ref{multstateqn}), the methods of the main body of the paper yield
\begin{align*}
\sum_{j=0}^{l} A_{0j} \mathbbm{1}_{t-j} \tilde{G}_{t-j} + \sum_{i=1}^h \sum_{j=0}^{l} A_{ij} \tilde{F}_{i,t-j}  = B R^t\ ,\  \forall t \geq 0\ .
\numberthis \label{multfgrecur} 
\end{align*} 
Application of equation~(\ref{fig}) leads to the following difference equation in $\/G_t\/$:  
\begin{align*}
\sum_{i=0}^h \sum_{j=0}^l A_{ij} \mathbbm{1}_{t-j} \tilde{G}_{t+i-j} = B R^t \ ,\ \forall t \geq 0\ .  
\numberthis \label{multgrecur} 
\end{align*} 

In order to ensure unique solutions of difference equations, irrespective of rational expectations, 
we shall assume that the corresponding matrix polynomial 
\begin{align*}
D[z] := \sum_{i=0}^h \sum_{j=0}^l z^{i+l-j} A_{ij} 
\end{align*} 
is regular.  
In that case, the model~(\ref{multstateqn},\ref{multinputeqn}) will also be called {\em regular\/}.  

Equation~(\ref{multgrecur}) has coefficients that vary with time, but only up to $\/t=l\/$, at the latest.  
The theory of time-invariant regular descriptor systems therefore implies that any solutions are of exponential order, 
and consequently possess z-transforms.  
Suppose that~(\ref{multgrecur}) and~(\ref{fig}), and therefore~(\ref{multfgrecur}), can be solved simultaneously.  
Taking z-transforms of both sides of~(\ref{fig}) yields, for any $\/0 < i \leq h\/$, 
\begin{align*}
\tilde{F}_{i}[z] = {\cal Z} \left \{ \mathbbm{1}_{t} \tilde{G}_{t+i} \right \} 
= {\cal Z} \left \{ \tilde{G}_{t+i} \right \} 
=  z^{i} \left [ \tilde{G}[z] - \sum_{k=0}^{i-1} z^{-k} \tilde{G}_k \right ] \ .  
\numberthis \label{fijzgz}
\end{align*} 
Then taking transforms in~(\ref{multgrecur}), and substituting according to~(\ref{fijzgz}), 
\begin{align*}
\sum_{i=0}^h \sum_{j=0}^l A_{ij}  z^{i-j} \left [ \tilde{G}[z] - \sum_{k=0}^{i-1} z^{-k} \tilde{G}_k \right ] = B [I - R z^{-1}]^{-1} \ .  
\numberthis \label{multfgrecurz}
\end{align*} 
By regularity, $\/\tilde{G}[z]\/$ must therefore satisfy 
\begin{align*} 
\tilde{G}[z] = & D[z]^{-1} z^l 
\left[ \sum_{i=1}^{h} \sum_{j=0}^l \sum_{k=0}^{i-1} A_{ij} z^{i-j-k} \tilde{G}_k + B [I - R z^{-1}]^{-1} \right ] \ .  
\numberthis \label{multgz}
\end{align*} 
Determination of a unique solution requires the specification of the initial values $\/\tilde{G}_i\/$, $\/0 \leq i < h\/$.  
These must be related to the initial values of the $\/\tilde{F}_{i,t}\/$, for $\/ 0 < i < h\/$, and to that of $\/A_{h0} \tilde{F}_{h,0}\/$, by the equations 
 $\/\tilde{G}_0 = B - \sum_{i=1}^h A_{i0} \tilde{F}_{i,0}\/$ (by~(\ref{multfgrecur}) and~(\ref{fig})), and $\/\tilde{G}_i = \tilde{F}_{i,0}\/$, $\/0 < i < h\/$ (by~(\ref{fig})).  
 Any two solutions $\/\tilde{G}_t\/$ that satisfied these initial values would have to have this z-transform, and to vanish for negative $\/t\/$.  
 It follows from z-transform inversion that the two solutions would in fact be equal.

\begin{proposition}
\label{multzsprop}
Suppose that the model~(\ref{multstateqn},\ref{multinputeqn}) is regular.  
For any possible values of $\/\tilde{F}_{1,0},\ \tilde{F}_{2,0},\ \ldots,\  \tilde{F}_{h-1,0}\/$ and $\/A_{h0} \tilde{F}_{h,0}\/$, 
define $\/\tilde{G}_0 := B - \sum_{i=1}^h A_{i0} \tilde{F}_{i,0}\/$ and $\/\tilde{G}_i := \tilde{F}_{i,0}\/$, for all $\/0 < i < h\/$.  
Let 
\begin{align*}
\tilde{N}[z] & := \left [ \sum_{i=1}^h \sum_{j=0}^l \sum_{k=0}^{i-1} A_{ij} z^{i+l-j-k} \tilde{G}_k + z^{l} B [I - R z^{-1}]^{-1}  \right ]\ .
\end{align*} 
Then there exists a solution of~(\ref{fig},\ref{multfgrecur}), 
consistent with the above initial values of the $\/\tilde{F}_{i,0}\/$, $\/0 < i < h\/$, and $\/A_{h,0} \tilde{F}_{h,0}\/$, 
and such that the $\/\tilde{F}_{i,t}\/$ and $\/\tilde{G}_t\/$ vanish for negative $\/t\/$,  
if and only if the following rational matrix is proper:  
\begin{align*}
\tilde{F}_h[z] = z^h \left [ D[z]^{-1} \tilde{N}[z] - \sum_{k=0}^{h -1} z^{-k} \tilde{G}_k \right ] \ .  
\end{align*}   
In that case, the unique such solution is given by 
\begin{align*}
\tilde{F}_{i,t}  & = {\mathcal Z}^{-1}\left \{ z^i \left [ D[z]^{-1} \tilde{N}[z] - \sum_{k=0}^{i -1} z^{ -k} \tilde{G}_k \right  ] \right \} \ , \ 0 < i \leq h\ , \\
\tilde{G}_t & = {\mathcal Z}^{-1} \left \{ D[z]^{-1} \tilde{N}[z] \right \} \ .   
\end{align*} 

A model-consistent forecasting mechanism~(\ref{multtildeforecasts}) exists if and only if the above condition is satisfied.  
Any such forecasting mechanism must have $\/\tilde{F}_{ij,t} = \tilde{F}_{i,t-j}\/$, for all $\/0 < i \leq h\/$, $\/0 \leq j \leq l\/$.  

If all initial conditions are zero-valued, and $\/ \hat{x}_{i,t-j} = \sum_{\tau = 0}^{t-j} \tilde{F}_{i,t-j-\tau} w_\tau\/$, for all $\/t \geq j\/$, and for all $\/ i,j,\ 0 < i \leq h,\ 0 \leq j \leq l\/$, 
then the model~(\ref{multstateqn},\ref{multinputeqn}) satisfies 
$\/x_t = \sum_{\tau = 0}^t \tilde{G}_{t-\tau} w_\tau\ ,\ \forall t \geq 0\/$.  
For any $\/i,j\/$, $\/0 < i \leq h\/$, $\/0 \leq j \leq l\/$, and any $\/t \geq j - i\/$, 
the forecast errors are given by:  
\begin{align*}
& x_{t+i-j} - \hat{x}_{i,t-j} = \sum_{\substack{\tau = \\ t-j+1}}^{t+i-j} \tilde{G}_{t+i-j-\tau} w_\tau\ .  
\numberthis \label{multerrformula}
\end{align*} 
\end{proposition}

\begin{proof} 
If a solution of~(\ref{fig},\ref{multfgrecur}) exists, then, by the above discussion, the unilateral z-transforms of the $\/\tilde{F}_{i,t}\/$ and $\/\tilde{G}_t\/$ must 
have the forms given in the statement of the proposition, and must be proper.  
This establishes the necessary condition for existence.  

For sufficiency, note that 
\begin{align*} 
\tilde{G}[z] & = \sum_{k=0}^{h-1} z^{-k} \tilde{G}_k + z^{-h} \tilde{F}_h[z]\ ; 
\end{align*} 
hence, if $\/\tilde{F}_h[z]\/$ is proper, so is $\/\tilde{G}[z]\/$.  
Both matrices are therefore unilateral z-transforms:  
their respective inverse transforms vanish for negative $\/t\/$.  
Moreover, the first $\/h\/$ values of the inverse transform of $\/\tilde{G}[z]\/$ 
are the respective $\/\tilde{G}_i\/$, $\/0 \leq i < h\/$.  
It follows that all the $\/\tilde{F}_i[z]\/$ are proper.  

By definition, $\/\tilde{G}[z]\/$ satisfies~(\ref{multgz}), and transforming to the time domain 
shows that its inverse transform $\/\tilde{G}_t\/$ satisfies~(\ref{multgrecur}).  
The $\/\tilde{F}_i[z]\/$ by definition satisfy~(\ref{fijzgz}), for $\/0 < i \leq h\/$, and transforming to the time domain 
shows that their respective inverse transforms satisfy~(\ref{fig}).  
It follows that~(\ref{multfgrecur}) is satisfied.  

Setting $\/t=0\/$ in~(\ref{fig}) implies that the initial values of 
the inverse transforms $\/\tilde{F}_i[z]\/$ for $\/0 < i < h\/$ are indeed the assigned values, 
and setting $\/t=0\/$ in~(\ref{multfgrecur}) shows that $\/A_{h0} \tilde{F}_{h,t}\/$ takes 
on its assigned value at $\/t=0\/$.  
This establishes the sufficiency of the condition for the existence of an appropriate solution of ~(\ref{fig},\ref{multfgrecur}).  
Uniqueness follows from the above discussion.  

The respective unique forms of the zero-state response of $\/\hat{x}_{i,t-j}\/$ and $\/x_t\/$ under a model-consistent forecasting mechanism then 
follow, by the previous discussion, and the expression for the errors by straightforward subtraction of convolution sums.  \qed
\end{proof} 

(As before, the forecasting mechanism can easily be expressed in terms of the $\/u_t\/$.)   

The parameters $\/\tilde{F}_{1,0}\/$, $\/\tilde{F}_{2,0}\/$\ \ldots $\/\tilde{F}_{h-1,0}\/$, and $\/A_{h0} \tilde{F}_{h,0}\/$ therefore determine $\/\tilde{G}_0\/$ 
via~(\ref{multstateqn}-\ref{multforecasts}); 
given these values and model equations, 
model-consistency then determines $\/\tilde{F}_{i,t}\/$, for $\/0 < i \leq h\/$, and $\/\tilde{G}_t\/$, for all $\/t > 0\/$. 

The sufficient ``well-posedness'' condition for existence of solutions generalizes as follows:  
\begin{corollary}
Suppose that the model~(\ref{multstateqn},\ref{multinputeqn}) is regular, and that \linebreak $\/D[z]^{-1} z^{h+l-1}\/$ is proper. 
Then for any possible values of $\/\tilde{F}_{1,0},\ \tilde{F}_{2,0},\ \ldots,\  \tilde{F}_{h-1,0}\/$ and $\/A_{h0} \tilde{F}_{h,0}\/$, 
there exists a (unique) model-consistent forecasting mechanism.  
\end{corollary} 
\begin{proof}
In accordance with Proposition~\ref{multzsprop}, define $\/\tilde{G}_0 := B - \sum_{i=1}^h A_{i0} \tilde{F}_{i,0}\/$ and $\/\tilde{G}_i := \tilde{F}_{i,0}\/$, for all $\/0 < i < h\/$. 
Note that
\begin{align*} 
\tilde{N}[z] := & \left [ \sum_{i=1}^h \sum_{j=0}^l \sum_{k=0}^{i-1} A_{ij}z^{i+l-j-k} \tilde{G}_k   + z^{l} B [I - R z^{-1}]^{-1}  \right ] \\ 
                  =\ & \left [ \sum_{i=0}^h \sum_{j=0}^l \sum_{k=0}^{h-1} A_{ij}z^{i+l-j-k} \tilde{G}_k \right . \\
                      & \left .  - \sum_{i=0}^{h} \sum_{j=0}^l \sum_{k=i}^{h-1} A_{ij}z^{i+l-j-k} \tilde{G}_k + z^{l} B [I - R z^{-1}]^{-1}   \right ] \\ 
                  =\ & D[z] \sum_{k=0}^{h-1} z^{-k} \tilde{G}_k \\ 
                      & - z^l \left [\sum_{i=0}^{h} \sum_{j=0}^l \sum_{k=i}^{h-1} A_{ij}z^{i-j-k} \tilde{G}_k (1 {-} \delta_j \delta_{k-i}) {-} B R z^{-1} [I {-} R z^{-1}]^{-1}  \right ]
\end{align*} 
where $\/\delta_t\/$ is the Kronecker delta function, whose value is unity when $\/t=0\/$, but otherwise vanishes.  Consequently, 
\begin{align*}
\tilde{F}_h[z] =\ & D[z]^{-1} z^{h+l} \times \\
                         &  \left [ B R z^{-1} [I - R z^{-1}]^{-1} - \sum_{i=0}^{h} \sum_{j=0}^l \sum_{k=i}^{h-1} A_{ij}z^{i-j-k} \tilde{G}_k (1 - \delta_j \delta_{k-i})  \right ] \ .  
\end{align*} 
The term in square brackets is strictly proper, so, if $\/D[z]^{-1} z^{h+l-1}\/$ is proper, then $\/\tilde{F}_h[z]\/$ must be proper.  \qed
\end{proof}

\vfill
\pagebreak

\section{Proofs for section~\ref{wellposed}}
\label{wellposedproofs}

The following simple lemma lists some of the implications of 
well-posedness.  

\begin{lemma}
\label{propercondn}
Suppose that $\/\hat{A}\/$ is nonzero and $\/[z^2 \hat{A} - zI + A]\/$ is regular.  
Then the following are equivalent:  
\begin{align*}
    & [z^2 \hat{A} - z I + A]^{-1} \ \text{is strictly proper} \\
\iff & [z^2 \hat{A} - z I + A]^{-1} z I \ \text{is proper} \\
\iff &  [z^2 \hat{A} - z I + A]^{-1} [z I - A] \ \text{is proper} \\
\iff & [z^2 \hat{A} - zI + A]^{-1} z^2 \hat{A} \ \text{is proper} \\
\iff & [z^2 \hat{A} - zI + A]^{-1} z \hat{A} \ \text{is strictly proper}\\
\iff & [z \hat{A} -  I]^{-1}\ \text{is proper} \\
\iff & z \hat{A} [z \hat{A} -  I]^{-1} \text{ is proper.}  
\end{align*} 
\end{lemma}
\begin{proof} 
The first four equivalences and the final one are straightforward.  
For the fifth, note that 
$\/[z \hat{A} - I]^{-1}\/$ can be realized from $\/[z \hat{A} - I + A z^{-1}]^{-1}\/$, 
and vice versa, 
by feedback through $\/A z^{-1}\/$:   
that implies that one is proper if and only if the other is.  \qed
\end{proof}

As claimed, well-posedness ensures the existence of model-consistent forecasting mechanisms:  

\begin{proposition}
\label{structure} 
Suppose that 
the model~(\ref{stateqn},\ref{inputeqn}) 
is regular and well-posed, 
and that the initial conditions are weakly consistent.  
Then there exists a (unique) model-consistent forecasting mechanism for~(\ref{stateqn},\ref{inputeqn}) for any 
given value of $\/\hat{A}F_0\/$.  
\end{proposition} 
\begin{proof}
Consider that $\/\overline{X}[z]\/$ is 
\begin{align*} 
&  [z^2 \hat{A} - zI + A]^{-1} [z^2 \hat{A} \hat{x}_{1,-1} - z A x_{-1} - z B R [I - R z^{-1}]^{-1} u_{-1}] \\
 =\  & \hat{x}_{1,-1} + [z^2 \hat{A} - zI + A]^{-1} [[zI - A] \hat{x}_{1,-1} - z A x_{-1} - z B R [I - R z^{-1}]^{-1} u_{-1}] \\
 =\  & \hat{x}_{1,-1}\\ 
    & + [z^2 \hat{A} - zI + A]^{-1} [ z (\hat{x}_{1,-1} - A x_{-1} - BR u_{-1}) \\ 
    & \ \ \ \ \ \ \ \ \ \ \ \ \ \ \ \ \ \ \ \ \ \ \ \ \ \ \ \  - A \hat{x}_{1,-1} - BR^2 [I - R z^{-1}]^{-1} u_{-1}] \ .
\end{align*} 
The weak consistency of the initial conditions implies that the term in parentheses lies within the image of $\/\hat{A}\/$;  
if, in addition, $\/[z^2 \hat{A} - z I + A]^{-1}\/$ is strictly proper, then by Lemma~\ref{propercondn}, $\/\overline{X}[z] - \hat{x}_{1,-1}\/$ also  
is strictly proper. 

Now write 
\begin{align*}
{F}[z] & = [z^2 \hat{A} - z I + A]^{-1} [[zI - A](\hat{A}\tilde{F}_0 + B)[zI - R] - z^2 B]  \\
         & = [z^2 \hat{A} - z I + A]^{-1} \left [ z^2 \hat{A} \tilde{F}_0 \right . \\
         & \ \ \ \ \ \ \ \ \ \ \ \ \ \ \ \ \ \ \ \ \ \ \ \ \ \ \ \  \left . - \left ([zI - A] (\hat{A} F_0 + B) R + z A (\hat{A} F_0 + B) \right ) \right ] \/ .  
\end{align*} 
If $\/[z^2 \hat{A} - z I + A]^{-1}\/$ is strictly proper,
then by Lemma~\ref{propercondn}, this is a sum of proper rational matrices, so $\/F[z]\/$ is proper.  
By Theorem~\ref{exun} therefore, there exists a unique  
model-consistent forecasting mechanism, regardless of the value of 
$\/\hat{A} \tilde{F}_0\/$. \qed
\end{proof}

This simple sufficient condition for the existence of model-consistent forecasting mechanisms 
also ensures the existence of realizations that incorporate feedback.\footnote{In the (usual) case where $\/A\/$ is nonzero.}    
Consider that the derivation of of $\/\tilde{F}[z]\/$ and $\/\tilde{G}[z]\/$ in 
section~\ref{zsresp} could have begun with the following system of equations, 
equivalent to~(\ref{gfrecur},\ref{onestep}):  
\begin{align*} 
\tilde{F}_t & = A \tilde{G}_t + \hat{A} \tilde{F}_{t+1} + B R^{t+1}\ , \ \forall t \geq 0\ , \\
\tilde{G}_{t} & = A \tilde{G}_{t-1}+ \hat{A} \tilde{F}_{t}  + B R^{t}\ , \ \forall t \geq 0\ .  
\end{align*} 

This yields the transformed equation 
\begin{align*} 
\setlength\arraycolsep{0pt}
\begin{bmatrix} 
I & - [I {-} z \hat{A}]^{-1} A \\
- [I {-} A z^{-1}]^{-1}\hat{A}  &  I 
\end{bmatrix} 
\begin{bmatrix}
\tilde{F}[z] \\ 
\tilde{G}[z] 
\end{bmatrix} 
{=}
\begin{bmatrix}
[I {-} z \hat{A}]^{-1} \left [B R [I {-} R z^{-1}]^{-1} {-} z \hat{A} \tilde{F}_0 \right ] \\
[I {-} A z^{-1}]^{-1} B [I {-} R z^{-1}]^{-1} 
\end{bmatrix}
\numberthis \label{feedback}
\end{align*} 
The left-hand side represents a feedback interconnection, and the right-hand 
side a vector of exogenous signals that serve as inputs to the feedback loop.  
The inverse of the left-hand coefficient exists:  
\begin{align*}
& \begin{bmatrix}
- [z^2 \hat{A} {-} zI {+} A]^{-1} z & 0 \\
0 & - [z^2 \hat{A} {-} zI {+} A]^{-1} z 
\end{bmatrix}
\begin{bmatrix}
I \/{-}\/ Az^{-1} & A \\
\hat{A} & I {-} z\hat{A} 
\end{bmatrix} 
\begin{bmatrix}
I {-} z \hat{A} & 0 \\
0 & I {-} A z^{-1} 
\end{bmatrix} \  
\end{align*}
-- and, by Lemma~\ref{propercondn}, that inverse is proper if $\/[z^2 \hat{A} - z I + A]^{-1}\/$ is strictly proper 
(with the converse holding if $\/A\/$ is nonsingular); 
in that case, the product of the inverse with the right-hand side of the equation is also proper.   
The above equation therefore describes $\/\tilde{F}_t\/$ and $\/\tilde{G}_t\/$ as 
being uniquely and causally derived from each other, within a feedback loop, and from 
signals that are exogenous to that feedback loop.  
By linearity and time-invariance, the zero-state responses of 
$\/\hat{x}_{1,t}\/$ and $\/x_t\/$ inherit this relationship from their 
convolution kernels.  

More explicitly, a feedforward/feedback realization of the zero-state response can be obtained by transforming the first equation of~(\ref{feedback}) to the 
time domain, and convolving with the $\/w_t\/$ sequence:  
\begin{align*}
\hat{x}_{1,t} & = \sum_{\tau = 0}^{t} \varPhi_{t-\tau} [A x_\tau + B R u_\tau] - \sum_{\tau=0}^{t} \Psi_{t-\tau} \hat{A} F_0 w_\tau\ ,\ \forall t \geq 0\ .  
\end{align*} 
Here, $\/\varPhi_t := {\cal Z}^{-1}\{ [I - z \hat{A}]^{-1}\}\/$ and $\/\Psi_t := {\cal Z}^{-1}\{ [I - z \hat{A}]^{-1} z\}\/$.  
By Lemma~\ref{propercondn}, the first of these is the inverse transform of a proper rational matrix (whose product with $\/\hat{A}\/$ is strictly proper), and the second, the inverse transform of a matrix whose product with $\/\hat{A}\/$ is proper.  
Alternatively, $\/\Psi_t\/$ could be defined as the inverse transform of a proper matrix, $\/{\cal Z}^{-1}\{[I - z \hat{A}]^{-1} z \hat{A} \hat{A}^g\}\/$, 
where $\/\hat{A}^g\/$ is a generalized inverse of $\/\hat{A}\/$ (such that $\/\hat{A}\hat{A}^g \hat{A} = \hat{A}\/$).  

By setting $w_t = 0\/$ for all $\/t \geq 0\/$, and transforming to the z-domain, 
it is easy to check that the extended law in the statement of Theorem~\ref{feedbackthm} leads to the same zero-input response as found in section~\ref{ziresp} (provided that the initial conditions are weakly consistent).  
This establishes the theorem.  

The structure of the feedforward/feedback predictor is displayed in Figure~\ref{causaldiag}, where initial conditions are suppressed for simplicity, 
and the convolution kernels $\/\varPhi_t\/$ and $\/\Psi_t\/$ are represented by their z-transforms.  

\begin{figure}[ht] 
\begin{tikzpicture}[auto, node distance= 2cm,>=Latex] \matrix[column sep = .375cm, row sep = .300cm]
{
                     &                            & \node [block](errff){$[I - z \hat{A}]^{-1}\left [\hat{A}{F}_0 [zI - R] - B R  \right ] $}; & \node [sum](sum1){}; & \node [block](newcfb){$[I - z \hat{A}]^{-1}A$}; & & & \\
                     &                            &                                                                     & & & & \\
                     &                            &                                                                     &                                                    & & & & \\
\node [coordinate](d5){}; &         &                                                                     &                                                    & & & & \node [coordinate](d6){}; \\
                    &                             &                                                                     & \node [block](fcgain){$\hat{A}$}; & & & \\
\node (wt){}; & \node [dot](d2){}; & \node [block](exog){$B$}; & \node [sum](sum2){}; & & & \node [dot](d4){}; & \node (xt){}; \\
& & & & & & \\
& & & & \node [block](cfb){$A z^{-1}$}; & & & \\
};
\draw [->] (wt) -- node[pos=0.05,above] {$u_t$} (exog); 
\draw [->] (exog) -- node[pos=0.975,above] {$+$} (sum2);
\draw [->] (sum2) --  node[pos=0.965,above] {$x_t$} (xt); 
\draw [->] (d2) |- (errff); 
\draw [->] (errff) -- node[pos=0.85,above] {$-$} (sum1); 
\draw [->] (newcfb) -- node[pos=0.4,above] {} node[pos=0.8,above] {$+$} (sum1); 
\draw [->] (d4) |- (cfb); 
\draw [->] (d4) |- (newcfb); 
\draw [->] (sum1) -- node[pos=0.40] {$\hat{x}_{1,t}$} (fcgain); 
\draw [->] (fcgain) -- node[pos=0.85,right] {$+$} (sum2);
\draw [->] (d4) |- (cfb); 
\draw [->] (cfb) -| node[pos=0.95,left] {$+$} (sum2); 
\draw [dashed]  (d5) -- (d6); 
\end{tikzpicture}
\caption{Realization of model-consistent forecasts for well-posed model (with $\/\overline{x}_t \equiv 0\/$).}  
\label{causaldiag}
\end{figure}
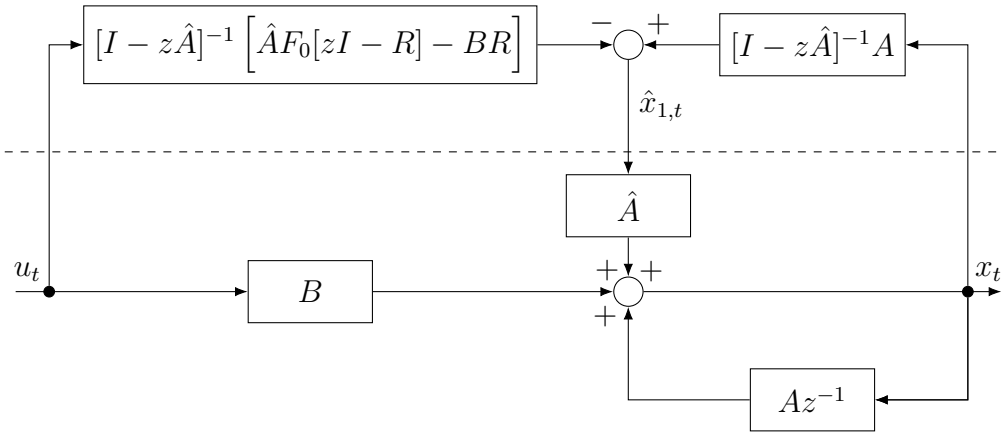 

\vfill
\pagebreak

\section{Mathematical preliminaries} 
\label{mathprelim}

This addendum briefly outlines some mathematical preliminaries relating to z-transforms and to polynomial and rational matrices.  
For more detail, refer for example to~\citet{Chen:linear} or to~\citet{FH:networks}.  

\subsection*{The unilateral z-transform}
\label{ztrans}

In employing frequency-domain methods to solve discrete-time initial-value problems, 
one generally applies the {\em unilateral\/}, or {\em one-sided\/} z-transform:  
\begin{align}
Y[z] := {\cal Z}\{y_t\} := \sum_{t=0}^\infty y_t z^{-t},\ z \in \mathbb{C}\ .  
\end{align}
Here, $\/y_t\/$ may be scalar-, vector-, or matrix-valued.  

A unilateral z-transform can be interpreted as that of the impulse response of a causal system.  
If it is a rational function (or a rational matrix\footnote{See the next section.}) it must therefore be {\em proper}:  
the numerator (of any of its elements) should have a degree no greater than that of the denominator.  

The transform is obviously linear; other fundamental properties are summarized below.  

\subsubsection*{Convergence} 

Suppose that every element $\/(y_t)_{ij}\/$ of the 
matrix $\/y_t\/$ satisfies $\/|(y_t)_{ij}| \leq K \alpha^t\/$, for some positive $\/K, \alpha \in \mathbb{R}\/$.  
Then the z-transform of $\/y_t\/$ converges wherever $\/|z| > \alpha\/$.  

When all of its elements satisfy such inequalities, $\/y_t\/$ is said to be of {\em exponential order\/}.  
Thus, polynomials and exponentials are of exponential order, as are sums, products, and convolutions of functions of exponential order.  

\subsubsection*{Inversion integral}  

The time-domain function $\/y_t\/$ is determined by $\/Y[z]\/$ via the 
following contour integral:  
\begin{align*}
y_t = \frac{1}{2 \pi i} \oint Y[z] z^{t-1} dz \  ,\ \forall t \in \mathbbm{Z}, 
\end{align*}
where the integration is performed in the counterclockwise direction around a closed contour 
within the region of convergence of the z-transform.  

If $\/Y[z]\/$ is a proper rational function (or a proper rational matrix -- see below), 
the inverse transform vanishes for negative $\/t\/$.  

In practice, inversion is often performed by other means than a direct evaluation of the above integral.  

\subsubsection*{Left-shift rule}

The transform, by definition, ignores any nonzero values of $\/y_t\/$ for negative values of $\/t\/$.  
It therefore always yields a transform whose inverse (see below) vanishes for negative values of $\/t\/$.  
This feature is reflected in the standard rule for left shifts of time-domain functions:  
\begin{align}
{\cal Z}\{y_{t+1}\} & := \sum_{t = 0}^\infty y_{t+1} z^{-t} 
                            = z \sum_{t=0}^\infty y_{t+1} z^{-(t+1)} 
                            = z [Y[z] - y_0] 
\end{align} 
The transform of the shifted sequence $\/y_{t+1}\/$ is obtained by simply multiplying the transform 
of the unshifted sequence $\/y_t\/$ by $\/z\/$ -- after annihilating the first element of the sequence, 
so that its left-shifted version vanishes for negative indices.  

More generally, we have by repeated application of the above, 
\begin{align*}
{\cal Z}\{y_{t + \tau}\} & = z^\tau \left [Y[z] - \sum_{k=0}^{\tau - 1} z^{-k} y_k \right ]\ ,\forall \tau \geq 1.  
\end{align*} 

We illustrate the definition and the shift operation by finding 
the unilateral z-transform of the sequence that is $\/R^t\/$ 
for nonnegative $\/t\/$, and zero otherwise.  
Apply a left shift after subtracting $\/R^0 = I\/$, which of course yields the same result as multiplying 
each term of the exponential sequence by $\/R\/$.  
The z-transform $\/R[z]\/$ of the original sequence therefore satisfies:  
\begin{align*}
z[R[z] - I] = R R[z]\ .  
\numberthis \label{zexpcalc} 
\end{align*} 
Solving, we find 
\begin{align*}
R[z] = [zI - R]^{-1} z I = [I - R z^{-1}]^{-1}\ .  
\end{align*} 
(because the matrix polynomial $\/[zI - R]\/$ is regular -- see the next section).  
The sum converges, and the z-transform exists,  if and only if $\/|z|\/$ is greater than the spectral radius 
of the matrix -- that is, the largest modulus of any eigenvalue.  

Applying~(\ref{zexpcalc}), we obtain an equality that is invoked implicitly in our calculations:  
\begin{align*} 
z [[I - R z^{-1}]^{-1} - I] = R [I - R z^{-1}]^{-1} 
\end{align*}

\subsubsection*{Right-shift rule}

The fundamental right-shift rule is as follows.  
\begin{align*}
{\cal Z}\{y_{t-1}\} 
= \sum_{t=0}^\infty y_{t-1}z^{-t}  
= z^{-1} [\sum_{\tau = 0}^\infty y_\tau z^{-\tau}  + z y_{-1}]
= z^{-1} [Y[z] + z y_{-1}]\ .
\end{align*}
By repeated application, we find, more generally, 
\begin{align*}
{\cal Z}\{y_{t - \tau}\} &= z^{-\tau}\left [Y[z] + \sum_{k = 1}^{\tau} z^{k} y_{-k} \right ] \  ,\forall \tau \geq 1.  
\end{align*} 

\subsubsection*{Time-domain convolution} 

\begin{align*}
{\cal Z}\{\sum_{\tau = 0}^t x_{t - \tau} y_\tau \} & = X[z] Y[z] \ .  
\end{align*} 

\subsection*{Polynomial and rational matrices}  
\label{matpoly}

An $\/n \times m\/$ {\em matrix polynomial\/} of degree $\/d\/$ is a polynomial with $\/n \times m\/$ matrix coefficients:  
\begin{align*}
P[z] = z^d A_d + z^{d-1} A_{d-1} + \ldots\ + A_0\  ,\ (A_d \neq 0)\ .  
\end{align*} 
Equivalently, an $\/n \times m\/$ {\em polynomial matrix\/} of degree $\/d\/$ is a matrix whose entries are polynomials of maximum degree $\/d\/$.  

A square matrix polynomial ($\/n \times n\/$) is {\em regular\/} 
if its determinant is not the zero polynomial.  
An eigenvalue is a value $\/\lambda\/$ of $\/z\/$ such that there exist nonzero 
vectors $\/x\/$ and $\/y\/$ for which 
\begin{align*}
P[\lambda] x = 0 = y^\top P[\lambda]\ .  
\end{align*} 
The vectors $\/x\/$ and $\/y\/$ are respectively right and left {\em eigenvectors\/} associated 
with $\/\lambda\/$.  
An $\/n \times n\/$ regular matrix polynomial has $\/nd\/$ eigenvalues, 
of which, in general, some are {\em infinite\/} (lying at the point at infinity on the Riemann sphere).  
If the leading coefficient $\/A_d\/$ is nonsingular, then all $\/nd\/$ eigenvalues are finite.  
if $\/A_d\/$ is singular, then every zero eigenvalue of $\/A_d\/$ gives rise to a zero 
eigenvalue of the {\em reverse\/} polynomial,  
\begin{align*}
z^ d [(z^{-1})^d A_d + (z^{-1})^{d-1}A_{d-1} + \ \dots \ + A_0] = z^d A_0 + z^{d-1} A_1 + \ \ldots \ + A_d 
\end{align*} 
The finite, nonzero eigenvalues of $\/P[z]\/$ are the reciprocals of those of the reverse polynomial; the infinite eigenvalues of $\/P[z]\/$ correspond to the zero eigenvalues of the reverse polynomial -- that is, to the zero eigenvalues of $\/A_d\/$.  
The degree $\/r\/$ of $\/{\text det} P[z]\/$ equals the number of finite eigenvalues of $\/P[z]\/$, counting multiplicities; the number of infinite eigenvalues of $\/P[z]\/$ is $\/nd - r\/$.  

A {\em left (resp., right) matrix fraction description\/} (MFD) is a representation of an $\/n \times m\/$ {\em rational matrix\/} -- a matrix 
of rational functions -- in the form 
$\/
D(z)^{-1} N(z) 
\/$ 
(resp., $\/N[z] D[z]^{-1}\/$), 
where $\/D[z]\/$ is a regular $\/n \times n\/$ (resp., $\/m \times m\/$) polynomial matrix (viewed here as a rational matrix, and hence possessing an inverse) and $\/N[z]\/$ an $\/n \times m\/$ polynomial matrix.  

Another common representation of a rational matrix has the form \linebreak 
$\/D_1[z]^{-1} N[z] D_2[z]^{-1}\/$, where $\/D_1[z]\/$ is regular and $\/n \times n\/$, $\/D_2[z]\/$ is regular and $\/m \times m\/$ and $\/N[z]\/$ is $\/n \times m\/$.  

A matrix of rational functions is {\em proper\/} (respectively, {\em strictly proper\/})  if each of its entries is proper (resp., strictly proper) -- that is, if, for each 
rational-function entry, the degree of the numerator polynomial is no greater than (resp., strictly less than) that of the denominator polynomial.  

If $\/z\/$ is interpreted as a left-shift operator (in accordance with the previous subsection), 
each of the rational-function entries represents a time-domain 
recurrence, with the numerator polynomial acting on an exogenous variable and the denominator on an 
endogenous variable, then properness implies nonanticipation:  the endogenous variable does not depend on future values 
of the exogenous variable.  
The converse also holds.\footnote{Indeed, an alternative definition of strict properness is that the Laurent expansion of the matrix contains no nonnegative powers of $\/z\/$ \citep{FH:networks}:  accordingly, we adopt for convenience the convention that an identically zero-valued matrix is strictly proper.}

Both properness and strict properness of rational matrices are preserved under addition, subtraction and multiplication.  
Moreover, multiplication of a proper matrix by a strictly proper one yields a strictly proper matrix.  

A (``negative'') feedback interconnection of two proper rational matrices $\/G[z]\/$ (in the forward path) and $\/H[z]\/$ (in the feedback channel), 
such that the product $\/G[z] H[z]\/$ is strictly proper, can be represented by the matrix 
\begin{align*} 
M[z] := [I + G[z] H[z]]^{-1} G[z]\ .  
\end{align*} 
The rational matrix $\/[I + G[z]H[z]]^{-1}\/$ is proper.  
It follows that $\/M[z]\/$ is proper; indeed, if $\/G[z]\/$ is strictly proper, then so is $\/M[z]\/$.  

\end{document}